\documentclass[letterpaper, 12pt, reqno, fleqn]{article}
\usepackage[top = 1in, bottom = 1in, left = 1in, right = 1in]{geometry}
\usepackage{lmodern}
\usepackage[T1]{fontenc} 
\usepackage{natbib}
\usepackage[figuresright]{rotating}
\usepackage{float}
\usepackage{graphicx}
\usepackage{epstopdf}
\usepackage{url}
\usepackage[colorlinks,citecolor=blue,linkcolor=blue,urlcolor=red,pagebackref]{hyperref}
\usepackage{amsmath,amsthm}
\usepackage[dvipsnames]{xcolor}
\usepackage[classfont = sanserif, langfont = roman, funcfont = italic]{complexity}
\usepackage{amssymb}
\usepackage[yyyymmdd,hhmmss]{datetime}
\usepackage{algpseudocode}
\usepackage{algorithm}
\usepackage{orcidlink}
\usepackage{dsfont}

\newtheorem{theorem}{Theorem}
\newtheorem{lemma}{Lemma}
\newtheorem{observation}{Observation}
\newtheorem{conjecture}{Conjecture}

\theoremstyle{definition}

\newtheorem{remark}{Remark}
\theoremstyle{remark}
\newtheorem{definition}{Definition}
\theoremstyle{open}
\newtheorem{open}{Open Problem}
\theoremstyle{informal}
\newtheorem{informal}{Informal Theorem}

\newcommand{\lp}{\left(}
\newcommand{\rp}{\right)}
\newcommand{\mc}{\mathcal}

\providecommand{\keywords}[1]
{
  \small	
  \textbf{Keywords:} #1
}

\newclass{\SHARPP}{\#P}
\newclass{\PPOLY}{P/Poly}
\newclass{\ETH}{ETH}
\newclass{\SLNP}{SLNP}

\newlang{\OV}{OV}
\newlang{\HAMPATH}{HAMPATH}
\newlang{\HAMCYCLE}{HAMCYCLE}
\newlang{\CLIQUE}{CLIQUE}
\newlang{\INDSET}{INDSET}
\newlang{\VERTEXCOVER}{VERTEX COVER}
\newlang{\COL}{COLOR}
\newlang{\CSP}{CSP}

\begin{document}

\title {Recovery Reductions, Conjectures, and Barriers}                      

\author{Tejas Nareddy\orcidlink{0009-0007-7032-6654}\footnote{Khoury College of Computer Science, Northeastern University. Email: \texttt{nareddy.s@northeastern.edu}.} \and 
Abhishek Mishra\orcidlink{0000-0002-2205-0514}\footnote{Department of Computer Science and Information Systems, Birla Institute of Technology and Science, Pilani, Pilani-333031, Rajasthan, I\textsc{ndia}. Email: \texttt{abhishek.mishra@pilani.bits-pilani.ac.in}.}}

\maketitle
\thispagestyle{empty}

\begin{abstract}

We introduce and initiate the study of a new model of reductions called the \textit{random noise model}. In this model, the truth table $\mathtt{T}_f$ of the function $f$ is corrupted on a randomly chosen $\delta$-fraction of instances. A randomized algorithm $\mathcal{A}$ is a $\left(t, \delta, 1-\varepsilon\right)$-recovery reduction for $f$ if:

\begin{enumerate}

\item With probability $1-\varepsilon$ over the choice of $\delta$-fraction corruptions, given access to the corrupted truth table, the algorithm $\mathcal{A}$ computes $f(\phi)$ correctly with probability at least $2/3$ on every input $\phi$.
  
\item The algorithm $\mathcal{A}$ runs in time $O(t)$.
  
\end{enumerate}

This model, a natural relaxation of average-case complexity, has practical motivations and is mathematically interesting.

Pointing towards this, we show the existence of robust deterministic polynomial-time recovery reductions with optimal parameters up to polynomial factors (that is, deterministic $\left( poly(n), 0.5 - 1/poly(n), 1-e^{-\Omega(poly(n))} \right)$-recovery reductions) for a large function class $\SLNP^S$ containing many of the canonical $\NP$-complete problems - $\SAT$, $k\SAT$, $k\CSP$, $\CLIQUE$ and more. As a corollary, we obtain that the barrier of \cite{Bogdanov2006} for non-adaptive worst-case to average-case reductions does not apply to our mild non-adaptive relaxation.

Furthermore, we establish recovery reductions with optimal parameters for \textit{Orthogonal Vectors} and \textit{Parity $k$-Clique} problems. These problems exhibit structural similarities to $\NP$-complete problems, with \textit{Orthogonal Vectors} admitting a $2^{0.5n}$-time reduction from $k\SAT$ on $n$ variables; and \textit{Parity $k$-Clique} a subexponential-time reduction from $3\SAT$.

\end{abstract}

\newpage

\pagenumbering{roman}
\setcounter{page}{1}

\maketitle

\tableofcontents

\newpage

\pagenumbering{arabic}
\setcounter{page}{1}

\section{Introduction}
\label{section:introduction}

Average-case complexity deals with the complexity of solving computational problems on, say, at least an 80\% fraction of instances of every size.  Much work has been done in the area, showing strong results for $\SHARPP$ and polynomials believed not to be in $\P$ \citep{Levin1986, Gemmell1992, Feige1996, Sudan1996, Cai1999, Sudan2001}. Typically, the paradigm to prove the average-case hardness of a problem is to show a reduction from solving the problem in the worst case to solving the problem on average. Recently, owing to the explosion of interest in fine-grained complexity theory \citep{Abboud2014, Williams2018, Williams2019}, work has also been focused on proving the fine-grained average-case hardness of problems in $\P$ \citep{Ball2017, Goldreich2018, Boix2019, Kane2019, Goldreich2020, Mina2020, Asadi2022, Asadi2024}.

While much progress has been made on the average-case hardness of non-Boolean functions, proving the average-case hardness of computing Boolean functions poses a significant technical challenge. In the absence of large finite fields to compute our function over, we may no longer straightforwardly use tools such as the Schwartz-Zippel-DeMillo-Lipton lemma \citep{Schwartz1980, Zippel1979, Demillo1978}. Hence, it is a natural question to ask what natural relaxations of average-case hardness might simultaneously be practically interesting, mathematically rich, applicable, and easier to prove for larger classes of functions (including Boolean functions). We propose a model that satisfies these criteria and provides a new approach to tackling the gap left by worst-case to average-average case reductions: the random noise model.

\subsection{The Setup: Using Corrupted Truth Tables}
\label{section:randomnoise}

We propose our model of ``almost'' worst-case to average-average case reductions. First, we define what corruption means in the random noise model.

\begin{definition}
\label{def:randomnoise}

\textbf{Random Noise Corruption}\\
Suppose $f_n:\Sigma^{p(n)}\to \mathcal{D}$\footnote{Here, we allow $\Sigma$ to be any constant length alphabet, and $\mathcal{D}$ is any set. The variable $n$ is the growing instance size parameter and $p:\mathbb{N}\to\mathbb{N}$ is a polynomial of constant degree.} is the part of the function $f$ of input length parameter $n$ and let $\mathtt{T}_{f_n}$ denote its truth table on instances of size parameter $n$ (length $p(n)$)\footnote{We use the phrase ``truth table'' even if $\mathcal{D}$ is not $\{\, 0,1 \,\}$, simply out of convention. When we say the truth table, we always refer to the table of all evaluations.}.

We define $\mathcal{N}_{\delta}:\mathcal{D}^{\left|\Sigma\right|^{p(n)}}\to\mathcal{D}^{\left|\Sigma\right|^{p(n)}}$ as a random noise operator that acts as follows, acting on $\mathtt{T}_{f_n}$ as $\mathcal{N}_{\delta}\mathtt{T}_{f_n}$ - a subset $S \subset \Sigma^{p(n)}$ of size $\delta|\Sigma|^{p(n)}$ is chosen uniformly at random. A corrupted truth table $\mathtt{T}_{f_n}^{\prime}$ is produced by modifying all entries in $S$ in any possible way, with no restrictions on how they are changed\footnote{These changes may be adversarial, may be made to minimize the time complexity of computing the function represented by $\mathtt{T}^{\prime}_{f}$, or to fit any other criterion.}, leaving all other entries unchanged.

\end{definition}

Now, we are ready to define what it means for a function $f$ to have a reduction in the random noise model.

\begin{definition}
\label{def:reduction}

\textbf{Recovery Reductions in the Random Noise Model}\\
For a given function $f:\Sigma^{*}\to \mathcal{D}$, a $\left(t(n), \delta, \varepsilon\right)$-recovery reduction in the random noise model is a randomized algorithm $\mathcal{A}$, defined as follows:

\begin{itemize}

\item The random noise operator $\mathcal{N}_{\delta}$ is applied on $\mathtt{T}_{f_n}$ to produce a corrupted truth table $\mathtt{T}_{f_n}^{\prime}$. It is unknown to the algorithm which answers are corrupted.
  
\item With probability $1-\varepsilon$ over the randomness of choice of corruptions, $\mathcal{A}$, given oracle access to the $\delta$-fraction corrupted truth table $\mathtt{T}^{\prime}_{f_n}$, computes $f_n(x)$ correctly in $O\left(t(n)\right)$ time with probability at least $2/3$.
  
\end{itemize}
\end{definition}

We emphasize that the algorithm must be correct on every inputs with probability at least $2/3$, conditioned on the $1-\varepsilon$ measure favorable corruption due the noise operator.

We use the phrase ``recovery reduction'' to highlight that such an algorithm both resembles the practical recovery algorithms that operate on corrupted data and a self-reduction in some sense. Going forward, we may also call this a ``recovery algorithm'' when thinking of it this way seems intuitively helpful.

This is a natural, practical, and well-explored model in the context of learning from noisy examples \citep{Angluin1988, Bianchi1999, Kalai2003, Akavia2008} and error recovery in databases \citep{Zhang2020}. These recovery and learning algorithms are helpful even if they work with a probability arbitrarily close to $1$ instead of with a probability of exactly $1$. Moreover, since if the class $\PH$ does not collapse, $\NP$-complete problems cannot have polynomial-time non-adaptive worst-case to average-case reductions \citep{Feigenbaum1993, Bogdanov2006}, this is a slight relaxation that we can think of as allowing ``almost'' worst-case to average-case reduction that may exist non-adaptively for $\NP$-complete problems. As we will see in later sections, this is a natural and unifying model for the significant $\NP$-complete problems and for some problems in $\P$ that are structurally similar to $\NP$-complete problems or admit fine-grained reductions from $\NP$-complete problems.

If $\varepsilon = 0$, the existence of a recovery reduction in this model is equivalent to $f$ having an $O\left(t(n)\right)$-time worst-to-average case reduction with error tolerance $\delta$. Constructing a reduction in the random noise model with $\varepsilon > 0$ does not strictly imply a hardness result for computing $f$ on a $\left(1-\delta\right)$-fraction of instances. Indeed, we prove that there is a $\left(\poly(n), 0.5-1/\poly(n), \exp\left(-\Omega(n^2)\right)\right)$-reduction for the decision problem of detecting cliques of size $n/2$, while it is known due to \citep{Erdos1963, Polya1937} that $\mathcal{G}(n, 1/2)$\footnote{The distribution of graphs on $n$ vertices where each edge is selected with probability $1 / 2$} does not have a clique of size $n/2$ with probability $1-o(1)$ - printing $0$ without reading the input is a good average-case algorithm. However, this does prove that $\left(1-\varepsilon\right)$-fraction of functions exactly $\left(1-\delta\right)$-close to $f$ are at least as hard to compute as the worst-case complexity of $f$.

\begin{remark}
\label{remark:relaxation}

Allowing $\varepsilon$ to be larger than $0$ is a natural parameter relaxation. For $\NP$-complete problems, since we still want to have polynomial time reductions we may either relax the worst-to-average-case reductions by allowing the error tolerance $\delta$ to be small or $\varepsilon$ to be non-zero. The result of \cite{Bogdanov2006} prohibits a relaxation of only $\delta$ to $1/\poly(n)$ unless $\PH$ collapses. This suggests that the most natural parameter to relax is $\varepsilon$.

We do note that one potential workaround to the barrier posed by \cite{Bogdanov2006} is to have an adaptive reduction. There exist many works showing the power of adaptive worst-to-average case reductions compared to non-adaptive ones \citep{Feigenbaum1992, Naik1993, Babai1999, Akavia2006}, and also landmark adaptive reductions \citep{Micciancio2004, Ajtai1996}.

\end{remark}

\subsection{Towards a Generalized $\NP$ Function with Symmetry}
\label{section:definitions}

Before we define the function class we construct recovery reductions for, we must define some preliminary algebraic structures, the first of which is a semilattice.

\begin{definition}
\label{def:semilattice}

\textbf{Semilattice}\\
A semilattice $\lp\mc{D}, \odot\rp$ is a set $\mc{D}$ and a binary operation $\odot: \mc{D} \times \mc{D} \to \mc{D}$ such that:

\begin{enumerate}

\item \textbf{(Associativity)} For every $x, y, z \in \mc{D}$, $\lp x \odot y\rp \odot z = x \odot \lp y \odot z \rp$.
  
\item \textbf{(Commutativity)} For every $x, y \in \mc{D}$, $x \odot y = y \odot x$.
  
\item \textbf{(Idempotence)} For every $d \in \mc{D}$, $d \odot d = d$.
  
\end{enumerate}

\end{definition}

Of course, we require that the semilattice operations be efficiently computable, so we define the following subclass of semilattices.

\begin{definition}
\label{def:computableSL}

\textbf{Polynomial-Time Computable Semilattice}\\
A semilattice $\lp\mc{D}, \odot\rp$ is polynomial time computable if for any elements $x$ and $y$ of $\mc{D}$, the product $x \odot y$ is computable in $\poly\lp\left|x\right|+\left|y\right|\rp$ time.
\end{definition}

As an exercise, one may verify that $\lp\{\, 0,1 \,\}, \vee\rp$ is a polynomial-time computable semilattice.

We now define a generalization of the class $\NP$ that, instead of simply the OR operation, allows any poly-time computable semilattice operation to be performed between the polynomial-time computable predicate values.

\begin{definition}
\label{def:aggregrate}

\textbf{Semilattice $\NP$}\\
Semilattice $\NP$ ($\SLNP$) is the class of functions defined as follows. A function $f$ is in the class $\SLNP$ if:

\begin{enumerate}

\item There is a polynomial $p: \mathbb{N} \to \mathbb{N}$ and a polynomial-time computable semilattice $\lp\mc{D}, \odot\rp$ such that for every natural number $n$, the function $f_n:\Sigma^{p(n)}\to \mathcal{D}$ is the function $f$ restricted to inputs of size parameter $n$ on a constant sized alphabet $\Sigma$.
  
\item There exists a function $h_n:\Sigma^{p(n)} \times \mathcal{C}_n\to \mathcal{D}$ such that
  $$f_n(\phi) = \bigodot_{x \in \mathcal{C}_n}h_n(\phi, x),$$
where $|x| = \poly(n)$, $\mathcal{C}_n = \Pi^{\poly(n)}$ for some constant- sized alphabet $\Pi$ and $h_n(\phi, x)$ can be computed in time polynomial in $n$.

\end{enumerate}

\end{definition}

As we noticed before that $\left(\{\, 0,1 \,\}, \vee\right)$ is a semilattice, it is simple to see that $\NP \subset \SLNP$ (proof in Appendix \ref{appendix:A}).

\begin{lemma}
\label{lemma:CSAF}

The complexity class $\NP$ is contained in the function class $\SLNP$.

\end{lemma}

Now, we add the restriction that our function has some symmetries involved, plus some other convenient properties.

\begin{definition}
\label{def:syminv}

\textbf{$G$-Invariant Semilattice $\NP$}\\
Suppose $G = \lp G_1, G_2, \ldots \rp$ is an infinite sequence of groups. A function $f$ is in $G$-Invariant Semilattice $\NP$, $\SLNP^G$ if:

\begin{enumerate}

\item The function $f$ is contained in the function class $\SLNP$.

\item For each group $G_n$, there exist group actions $\alpha: G_n \times \Sigma^{p(n)} \to \Sigma^{p(n)}$ ($p$ is the size function associated with the function $f$) and $\beta: G_n \times \mathcal{C}_n \to \mathcal{C}_n$ such that:
  
\begin{enumerate}
      
\item For every $\phi \in \Sigma^{p(n)}$, $x \in \mathcal{C}_n$, and $g \in G_n$, we have that $$h_n\left(\alpha_{g}(\phi), \beta_{g}(x)\right) = h_n\left(\phi, x\right).$$
  
\item The group action $\beta$ partitions the set (the certificate space) $\mathcal{C}_n$ into $j(n) = \poly(n)$ distinct orbits $\mathcal{C}^{\beta}_{n, 1}, \mathcal{C}^{\beta}_{n, 2}, \ldots, \mathcal{C}^{\beta}_{n, j(n)}$.
  
\item In $\poly(n)$ time, it is possible to compute a list $\left(x_1, x_2, \ldots, x_{j(n)}\right) \in \mathcal{C}_n^{j(n)}$ such that for each $i \in \left[j(n)\right]$, $x_{i} \in \mathcal{C}^{\beta}_{n, i}$.
  
\item The group actions $\alpha$ and $\beta$ are computable in $\poly\left(n\right)$ time.
  
\end{enumerate}

\end{enumerate}

\end{definition}

\begin{lemma}
\label{lemma:isoinv}

\textbf{Isomorphism Invariance Property of $\SLNP^G$}\\
Given an indicator function $f \in \SLNP^G$, with group sequence $G_1, G_2, \ldots$ and group actions $\alpha$ and $\beta$, we have that for every element $g \in G_n$ and input $\phi \in \Sigma^{p(n)}$, $$f_n\left(\alpha_g(\phi)\right) = f_n(\phi).$$

\end{lemma}

Our final definition forces that our sequence of groups is a sequence of symmetric groups.

\begin{definition}
\label{def:symgroupsyminv}

\textbf{$S$-Invariant Semilattice $\NP$}\\
A function $f$ is the class $S$-Invariant Semilattice $\NP$, or $\SLNP^S$ if:

\begin{enumerate}

\item The function $f$ is in a function class $G$-Invariant Semilattice $\NP$, $\SLNP^G$.
  
\item The group sequence $G = \lp G_1, G_2, \ldots\rp = \lp S_{m(1)}, S_{m(2)}, \ldots \rp$ where $m: \mathbb{N} \to \mathbb{N}$ is a function whose value and computation time grow at most polynomial in its input.
  
\end{enumerate}

\end{definition}

\subsection{Our Results}
\label{section:results}

We state below, the main result of our paper, giving recovery reductions in the random noise model for many of the canonical $\NP$-hard problems.

\begin{theorem}
\label{theorem:resultsSICSAF}

For every $\epsilon > 0$, any function $f$ in the class $\SLNP^S$ has a fully deterministic $\left(\poly(n, 1/\epsilon), 0.5 - \epsilon, \exp\left(-\Omega(\poly(n))\right)\right)$-recovery reduction in the random noise model.

\end{theorem}

\begin{theorem}
\label{theorem:results}

For every $\epsilon > 0$, the following problems have fully deterministic\\ $\left(\poly\left(n, 1/\epsilon\right), 0.5-\epsilon, \exp\left(-\Omega(\poly(n))\right)\right)$-recovery reductions in the random noise model - $\SAT$, $k\SAT$, $k\CSP$, MAX-$k\CSP$, $\CLIQUE$, $\INDSET$, $\VERTEXCOVER$, $k\COL$,\\ $\HAMCYCLE$ and $\HAMPATH$.

\end{theorem}

We state this result more informally, in an algorithmic sense, for the example of Boolean Satisfiability.

\begin{informal}

Suppose we are given a truth table $\mathtt{T}^{\prime}_{\SAT}$ for $\SAT$ instances on $n$ variables such that a randomly chosen $\left(0.5-1/\poly(n)\right)$-fraction of the answers are flipped. There is a polynomial time deterministic procedure with access to $\mathtt{T}^{\prime}_{\SAT}$ that recovers $\SAT(\phi)$ for every formula $\phi$, with probability $1-2^{-\Omega\left(\poly(n)\right)}$ over the choice of corruptions.

\end{informal}

We emphasize that this probabilistic guarantee is not for each $\phi$ independently, but that with this probability guarantee, our algorithm works correctly for \textit{every} input $\phi$.

Our reduction can be seen as an efficient deterministic recovery algorithm for entries of truth tables in $\NP$-complete problems that fails with very low probability over the choice of corruptions. Our result can be seen as saying, in a database recovery view, ``$\NP$-complete truth tables have redundancy built in.''

In fact, a famed conjecture of \cite{Berman1997} would imply that every $\NP$-complete problem has some form of a recovery reduction with the same paramaters as the one we have shown for $\SLNP^S$.

\begin{conjecture}

\textbf{Berman-Hartmanis Conjecture \citep{Berman1997}}\\
Between any two $\NP$-complete languages, there exists a bijection that is a polynomial-time reduction computable in either direction.

\end{conjecture}

\begin{theorem}

If the Bermann-Hartmanis Conjecture is true, for every $\NP$-complete function (the indicator functions of $\NP$-complete langauges), there exists a partition of possible inputs, such that all have fully deterministic\\ $\left(\poly\left(n, 1/\epsilon\right), 0.5-\epsilon, \exp\left(-\Omega(\poly(n))\right)\right)$-recovery reductions in the random noise model for every $\epsilon > 0$.

\end{theorem}

\begin{proof}

For any $\NP$-complete indicator function $f$, the polynomial-time bijection $\mathcal{B}$ to $\SAT$ due to the conjecture of \cite{Berman1997} forces a partition of the input space of $f$ on the basis of the input length of the image of $f$ across the function $\mathcal{B}$. On this partition, we apply the recovery reduction for $\SAT$ implied by Theorem \ref{theorem:resultsSICSAF} and Theorem \ref{theorem:results}.

\end{proof}

Our recovery algorithm is, in one sense, optimal for non-adaptive algorithms. As stated previously, if we were to improve the probability that the procedure works over the choice of random corruptions to $1$, then even a randomized polynomial-time algorithm would imply that $\PH$ collapses to the third level due to the work of \cite{Feigenbaum1993} and \cite{Bogdanov2006}.

\begin{informal}

The barrier of \cite{Bogdanov2006} for non-adaptive worst-case to average-case reductions requires that the noise be adversarial and not random.

\end{informal}

If we were to raise the fraction of corruptions to $0.5$, then this would imply the inclusion of $\NP$ in $\BPP$ since we would be able to simulate the queries to the truth table $\mathtt{T}^{\prime}_{\SAT}$ using truly random bits. That is, the truth table received is as good as receiving an almost uniformly random string of that length. We are unable to increase the corruption fraction to $0.5 - 2^{-\poly(n)}$ since $2^{-\poly(n)}$-fraction advantages in correctness, cannot generally be exploited in polynomial time. Moreover, our recovery algorithm is fully deterministic rather than randomized.

Each of our recovery reductions for all the $\NP$-hard problems listed in Theorem \ref{theorem:results} follows from the fact that they are all in the class $\SLNP^S$.

This subtly points towards symmetry and invariance being a structural property of $\NP$-hardness. We discuss this further in Section \ref{section:open} on open problems.

\subsubsection*{Recovery Reductions for Fine-Grained Problems} 

We also give recovery reductions in the random noise model for the \textit{Orthogonal Vectors} ($\OV$) problem , deciding whether there is a pair of $d$-dimensional $0/1$ vectors in a list of $n$ such vectors whose dot product (over $\mathbb{R}$) is $0$, and \textit{Parity $k$-Clique}, the problem of computing the lowest order bit on the number of cliques of size $k$ in a simple $n$-vertex graph. The average-case complexity of these problems has been well-studied. \cite{Ball2017} and \cite{Mina2020} show average-case hardness for the low degree extension and a construction called the ``factored version'' of the $\OV$ problem, respectively. For the \textit{Parity $k$-Clique} problem, \cite{Goldreich2020}, improving upon the work of \cite{Boix2019}, showed that there is an $O\left(n^2\right)$-time worst-to-average-case reduction from computing \textit{Parity $k$-Clique} in the worst case to computing it correctly on a $\left(1-2^{-k^2}\right)$-fraction of instances.

In contrast, for both problems, we give recovery reductions in our model with optimal parameters, while not modifying either Boolean function.

\begin{theorem}
\label{thoerem:fg}

\begin{enumerate}

\item For every $\epsilon = 1/\polylog(n)$ and dimension $d$ at most $O\left(n^{1-\gamma}\right)$ for some $\gamma > 0$, we have a $\left(\tilde{O}(nd), 0.5-\epsilon, 1-2^{-nd}\right)$-recovery reduction for the $\OV$ problem.

\item For every constant $k > 0$ and $\epsilon = 1/\polylog (n)$, we have a $\left(\tilde{O}\left(n^2\right), 0.5-\epsilon, 1-2^{-{{n}\choose{2}}}\right)$-recovery reduction for \textit{Parity $k$-Clique}.

\end{enumerate}

\end{theorem}

Here, we use randomness in our reduction so we have a reduction time that is linear (up to polylogarithmic factors) in the input size, to avoid the large polynomial overhead our black-box algorithm gives us.

In some way, both these problems have structure similar to $\NP$-complete problems. $\OV$ admits a $2^{0.5n}$-time reduction from $k\SAT$ with $n$ variables \cite{Williams2005} and \textit{Parity $k$-Clique} is $\ETH$-hard \citep{Goldreich2018, Chen2006}. This suggests a relationship between reduction from $\NP$ even if the reduction is subexponential, and the existence of a recovery reduction. We note that our techniques do not imply random reductions for arbitrary polynomial-time computable functions. We further discuss this in Remark \ref{remark:multiplication}.

\begin{remark}

We prove our main theorem (Theorem \ref{theorem:results}) for the $\NP$-hard problems listed in the theorem statement. However, we emphasize that this is not an exhaustive list of $\NP$-hard problems whose recovery reductions follow from Theorem \ref{theorem:resultsSICSAF}. We believe that many others follow, even with short proof sketches, but we only list the popular $\NP$-hard problems for brevity.

\end{remark}

\subsection{Open Problems}
\label{section:open}

\subsubsection{$\NP$-Completeness and Recovery Reductions}

Since we showed recovery reductions for many of the natural $\NP$-complete problems, it is natural to wonder if the existence of a recovery reduction is an inherent property of $\NP$-completeness. We conjecture the following about $\NP$-complete problems.

\begin{conjecture}
\label{conjecture:weakrandom}

\textbf{Recovery Reductions for Every $\NP$-Complete Problem}\\
For every $\epsilon > 0$, every $\NP$-complete problem has a \\ $\left(\poly(n, 1/\epsilon), 0.5-\epsilon, \exp\left(-\poly(n)\right)\right)$-recovery reduction in the random noise model.

\end{conjecture}

We also strengthen this and make this conjecture for deterministic recovery reductions.

\begin{conjecture}
\label{conjecture:weakdeterministic}

\textbf{Deterministic Recovery Reductions for Every $\NP$-Complete Problem}\\
For every $\epsilon > 0$, every $\NP$-complete problem has a deterministic \\ $\left(\poly(n, 1/\epsilon), 0.5-\epsilon, \exp\left(-\poly(n)\right)\right)$-recovery reduction in the random noise model.

\end{conjecture}

\begin{open}
\label{conjecture:strong}

Is $\NP$, considering the indicator function of each language, contained in the class $\SLNP^S$?

\end{open}

Due to Theorem \ref{theorem:resultsSICSAF}, a positive answer to \ref{conjecture:strong} implies Conjecture \ref{conjecture:weakdeterministic}, which in turn implies Conjecture \ref{conjecture:weakrandom}.

The authors believe Conjecture \ref{conjecture:strong} to be true since heuristically, $\NP$-complete problems must be expressive and since they admit reductions from problems in $\NP$ with rich symmetries, it seems as though $\NP$-complete problems must retain some essence of the original symmetries. 

As discussed before, \cite{Berman1997} originally conjectured that between any-two $\NP$-complete problems, there is a bijection between the two languages that is polynomial-time computable on both sides. They call this a $p$-isomorphism. This conjecture would already imply deterministic recovery reductions for $\NP$-complete languages not on truth tables of fixed input length for that language, but a truth table on the input set that is the image of the $p$-isomorphism from (say) $\SAT$ instances on $n$ variables. Really, this is a relabelled $\SAT$ truth table, and we could compute the labels via the two-way $p$-isomorphism and use the recovery reduction for $\SAT$ as the main procedure.

Hence, we believe that our conjectures are important in the study of the structure of $\NP$-completeness.

\begin{open}
\label{open:finegrained}

\textbf{Recovery Reductions for Expressive Problems in $\P$}\\
Suppose $f$ is a problem in $\P$ of (randomized) time complexity $T(N)$ (on instances of $f$ of size $N$) with a $2^{o(n)}$-time reduction from $n$ variable instances of $3\SAT$. Does every such function $f$ have a $\left(t(N), 0.5 - \epsilon, 1 - \exp\left(-\Omega(\poly(N))\right)\right)$-recovery reduction in the random noise model such that $t(N) = o\left(T(N)\right)$ for every constant $\epsilon > 0$.

\end{open}
 
While $\OV$ and \textit{Parity $k$-Clique} do not necessarily have these properties, we raise the above question. Implicitly assuming the Exponential Time Hypothesis ($\ETH$)\footnote{The Exponential Time Hypothesis, $\ETH$ states that $3\SAT$ on $n$ variables requires $2^{cn}$ deterministic time to decide for some constant $c > 0$.} \citep{Impagliazzo2001, Impagliazzo2001b}, if a problem $f$ admits a subexponential time reduction from an $\NP$-complete problem $L$, must it contain the relevant symmetry conditioned on the fact that the language $L$ contains symmetry? If so, it is highly possible that $f$ contains enough symmetry to have a non-trivial\footnote{A trivial recovery reduction would simply use the algorithm that computes $f$ normally in the worst-case.} recovery reduction.

We ask what consequences the above proposed conjectures might have.

\begin{open}
\label{open:consequences}

\textbf{Consequences of Recovery Reductions for Every $\NP$-Complete Problem}\\
What conditional results in complexity theory can be shown assuming Conjecture \ref{conjecture:weakrandom}, Conjecture \ref{conjecture:weakdeterministic}, or Conjecture \ref{conjecture:strong}?

\end{open}

\subsubsection{Studying Correction in Various Models of Errors}

We can see errors as follows - a source of errors produces a string and XORs the string to the truth table. In our case, the source picks a random subset of inputs, with fixed Hamming weight and adversarially chooses the entries there and we try to decode with high probability over the behavior of the source. In the traditional setting in coding theory, the source introduces an adversarial string of fixed Hamming weight and the task is to decode any input with high probability.\\
\\
\textbf{Informal Question 1.} Can we construct efficient recovery reductions for functions when the error is produced by an arbitrary $\AC^0$ or an $\NC^0$ source?\\
\\
We believe the above question to be an intermediate between the two settings - where the former is easy to construct reductions for and the latter is difficult and in some cases prohibitive to construct reductions in the latter model.

\section{Intuition and Technical Overview}
\label{section:techoverview}

Suppose we have our function $f$ in the class $\SLNP^S$. The concrete function to keep in mind here is $\CLIQUE$, where $f_n:[n] \times \{\, 0,1 \,\}^{n \choose 2} \to \{ \, 0,1 \,\}$\footnote{This can be modified to be of the form $f_n:\{\, 0,1 \,\}^{p(n)} \to \mathcal{D}$ if we treat $[n]$ as $\{\, 0,1 \,\}^{\lceil\log n\rceil}$ and let $f_n$ be uniformly zero if the first $\lceil\log n\rceil$ bits encode a number larger than $n$.}, written as $f_n\left(k, H\right)$, is $1$ if and only if the simple $n$ vertex graph $H$ contains a clique of size $k$. 

This function is in the class $\SLNP^S$ since we can write $$f_n(k, H) = \bigvee_{T \subset [n]} \mathds{1}_{\text{The vertex set encoded by }T\text{ forms a clique in }H\text{ and }|T| = k},$$ where $\mathds{1}_{A}$ is the indicator function of $A$\footnote{That is, $\mathds{1}_{A}$ is $1$ if the assertion $A$ is true and $0$ if $A$ is false.}. A permutation $\pi$ of $S_n$ acts on the graph $H$ by permuting it according to the permutation $\pi$ and acts on the set $T$ by replacing every element $i \in T$ with $\pi(i)$.

\subsection*{Example - Sketching a Randomized Recovery Reduction for $\CLIQUE$}

Our recovery reductions rely on the following intuitive observations.

\begin{enumerate}

\item Symmetry - having a large automorphism group - makes computation easy.

\item Suppose our input $\phi$ has a lot of distinct instances isomorphic to it - it has a large orbit. The law of large numbers guarantees that if the truth table $\mathtt{T}_{f}$ is corrupted at random on $49\%$ of instances, then with high probability, the corrupted truth table has at least $50.5\%$ of isomorphic instances have correct answers in the corrupted truth table $\mathtt{T}^{\prime}_{f}$.

\end{enumerate}

The Orbit-Stabilizer Theorem (Lemma \ref{lemma:orbit-stabitilizer}), saying that $\left|\textit{Aut}_{G}(\phi)\right|\left|\textit{Orb}_{G}(\phi)\right| = |G|$ guarantees that at least one of the two conditions is true!

\subsubsection*{When the Graph is Highly Symmetric}

The first observation means that if our instance $\phi$ is highly symmetric - it's automorphism group is large - then, we can compress the proof verification process. For our function $f(k, H)$, suppose $H$ has a very large automorphism group. Intuitively, this means that the presence of one clique of size $k$ implies the existence of many others - we do not need to check whether there is a clique of size $k$ for each of these vertex sets $S$ - it suffices to check only one of them in this family. If there is a clique of size $k$ in this position encoded by $S$, we can return the answer $1$ for $f(k, H)$. Otherwise, we conclude immediately that many other possible vertex families of size $k$ - at least those in the orbit of $S$ when acted on by the automorphism group $\textit{Aut}(H)$ - do not contain cliques of size $k$ in the graph $H$. 

Using the Orbit-Stabilizer Theorem, our job, if we know that the index of $|\textit{Aut}(H)|$, $$n!/|\textit{Aut}(H)|,$$ is small - say $O\left(n^3\right)$ - is to list a set of the $O\left(n^3\right)$ right coset representatives (a right transversal) $U^{R}_{H}$ of the automorphism group $\textit{Aut}(H)$ and evaluate whether the vertex set $\beta_{\pi}\left(\{\, 1, 2, \ldots, k \,\}\right)$\footnote{In this case, there is only one orbit of sets of size $k$ under action from $S_n$. Typically, we would need to evaluate over every orbit. Of course, here, for $\CLIQUE$, we know a priori we only need to evaluate for one size - this will generally not be the case.} contains a clique of size $k$ in the graph $H$ for each $\pi$ in the right transversal $U^{R}_{H}$. We list each coset representative in the following way - instantiate a list containing only the graph $H$. Keep sampling random permutations $\pi \in S_n$ and add $\alpha_{\pi}(H)$ to the list if it is not already contained in the list. From the Orbit-Stabilizer Theorem, we know that there should be $n!/\left|\textit{Aut}(H)\right| = O(n^3)$-many distinct graphs and we stop this procedure once we have that many. This is a classic case of the coupon collector problem, and we can cover the space in $O\left(n^3\log n\right)$ randomized time \citep{Flajolet1992, Mitzenmacher2017}. Then, we complete our $O\left(n^3\right)$ evaluations to compute $$\bigvee_{\pi \in U^{R}_{H}}\mathds{1}_{\beta_{\pi}\left(\{\, 1, 2, \ldots, k \,\}\right)\text{ is a clique of size }k\text{ in the graph }H}$$ and obtain a polynomial-time randomized algorithm to compute $\CLIQUE$ on highly symmetric graph inputs - with $O(n^3)$ index.

\subsubsection*{When the Graph is Less Symmetric}

The second observation, on the other hand is that the operator $\mathcal{N}_{0.49}$ acts on the truth table $\mathtt{T}_{f}$ to return the corrupted truth table $\mathtt{T}^{\prime}_{f}$ with $49\%$ of the entries randomly corrupted. Due to the law of large numbers, and formalized by concentration bounds, all large orbits - of size $\Omega(n^3)$ - of $\mathcal{P}\left([n]\right)$\footnote{This is the power set of $[n]$ in this case.} will have at least $50.5\%$ of their answers correct with high probability (over the randomness of the choice of corruptions). We can simply sample random permutations from $S_n$ and query the corrupted truth table on $\texttt{T}^{\prime}_{f}$ on polynomially many entries isomorphic to $H$. With high probability, the majority value is correct! With high probability over the randomness of corruptions, we now have a randomized polynomial-time query algorithm for the asymmetric case - when the index is $\Omega\left(n^3\right)$.

\subsubsection*{Distinguishing between the Two Cases}

Now, all that remains to be done is to distinguish between the two cases. We have already hinted at a method to compute the index of a subgroup. We noted before that we can sample random permutations $\pi \in S_n$ and list every new graph $\alpha_{\pi}(H)$ that we have not seen before. Suppose the index is indeed of size $O\left(n^3\right)$. Then not only is this list of size $O\left(n^3\right)$, but after $O\left(n^3\log n\right)$ samples, this list stops growing. To be safe, we may even use $O\left(n^4\right)$ samples and ensure this list actually stops growing. After $O\left(n^4\right)$-time, if our list stops growing and is below our $O\left(n^3\right)$ size threshold, we use our procedure for highly symmetric graphs. Otherwise, we use the procedure where we query the corrupted truth table $\mathtt{T}^{\prime}_{f}$.

\subsection*{Derandomizing the Recovery Reduction}

Now that we have an intuitive feel for how and why our algorithm works, our job remains to give deterministic procedures for the parts of the recovery algorithms that we used randomness in. Namely, we used randomness in our algorithms for the following subroutines:

\begin{enumerate}

\item Distinguishing between large and small index automorphism groups and listing coset representatives when the index of the group is small.

\item Querying the corrupted truth table $\mathtt{T}^{\prime}_{f}$.

\end{enumerate}

For the first case, luckily for us, much work has been done in the area of computational group theory, for deterministic algorithms for permutation groups. We use the following result.

\begin{lemma}
\label{lemma:SS}

\textbf{\citep{Furst1980, Sims1970}}\\
Suppose we have a positive integer $m$, and a subgroup $G$ of $S_m$ such that we have a deterministic membership test to determine if any element $\pi \in S_m$ is a member of $G$ in time $T_{G} = \Omega\left(m \log m\right)$. Then,

\begin{enumerate}

\item In $\poly(m)T_{G}$ deterministic time, it is possible to compute the size of the subgroup $G$, $|G|$.

\item For any integer $k$, in $\poly(m)\poly(k)T_{G}$ deterministic time, it is possible to print a list of $k$ distinct (left or right) coset representatives of $G$.\footnote{In standard texts, the algorithm runtime depends on the index of $G$ as $\poly\left(\left[S_m:G\right]\right)$. However, to print only $k$ coset representatives, it suffices to terminate this procedure after we have listed $k$ coset representatives, which is possible with a polynomial time dependence on $k$.}

\end{enumerate}
\end{lemma}

Using these algorithms and the fact that simply checking if $\alpha_{\pi}(H) = H$ is an efficient membership test, we can both determine the size of the automorphism group $\textit{Aut}(H)$ deterministically in polynomial-time, and compute a set of right coset representatives efficiently if we know the index is polynomial in $n$.

Now, derandomizing the second part is more challenging. A naive approach is to query the corrupted truth table $\mathtt{T}^{\prime}_{f}$ on every possible graph isomorphic to $H$, but the not-so-subtle issue for almost all graphs is that they have automorphism groups of size $1$ and we would be forced to make $n!$ queries \citep{Polya1937, Erdos1963}.

Suppose we use the deterministic right coset listing algorithm in Lemma \ref{lemma:SS} and construct a list $L$ of polynomially many right coset representatives of the automorphism group $\textit{Aut}(H)$  - can we query the truth table $\mathtt{T}_{f}^{\prime}$ on $\alpha_{\pi}(H)$ for every $\pi \in L$, take the majority value and hope we are correct?

As it turns out, we show using Chernoff bounds that indeed, with probability $1-\exp\left(-\poly(n)\right)$ over the randomness of the operator $\mathcal{N}_{0.49}$, for every graph $H$ with $\Omega(n^3)$ automorphism index size, the ``query path'' described above has at least $50.5\%$ of the answers correct within the corrupted truth table $\mathtt{T}^{\prime}_{f}$ under the assumption that $L$ is sufficiently large (but still polynomial). Hence, we can indeed simply do this and take the majority answer!

With these two ingredients, we have fully derandomized the recovery reduction for $\CLIQUE$. Our generalized recovery algorithm for $\SLNP^S$ proceeds almost identically.

\section{Preliminaries}
\label{section:preliminaries}

\subsection{Notation}

Throughout, we use $\mathtt{T}_{f_n}$ to refer to the truth table of the function section $f_n$ (of input length parameter $n$). We use $\mathtt{T}^{\prime}_{f_n}$ to refer to a corrupted truth table of the function section $f_n$. We use $\Sigma$ (outside of summation notation) to denote an alphabet of constant size. Typically, when we say $p(n)$, we refer to a positive function $p$ that is bounded from above by $n^{c}$ for some constant $c$. The notation $\poly(n)$ is a substitute for a positive function bounded from above by a polynomial in $n$ ($n^c$ for some constant $c$) and $\polylog(n)$ is a substitute for a positive function bounded from above by $\left(\log n\right)^c$ for some constant $c$. The notation $\ln$ represents the natural logarithm with base $e$ and $\log$ represents the logarithm with base $2$. For any positive integer $m$, $[m]$ denotes the set $\{\, 1, 2, \ldots m \,\}$. The notation $\exp(y)$ denotes $e^{y}$.

Throughout, $\alpha_g$ typically represents the group action of the element $g$ applied on the input to the function $f$ and $\beta_{g}$ represents the action of the element $g$ on a member $x$ of the certificate set $\mathcal{C}$. When $G^{\prime}$ is a subgroup of a group $G$, the index $\left[G:G^{\prime}\right] = |G|/ \left|G^{\prime}\right|$. $\textit{Aut}_{G}(y)$ represents the automorphism group of the object $y$ when acted on by the group $G$ and $\textit{Orb}_{G}(y)$ represent the orbit of $y$ when acted on by the group $G$. Throughout, we use $U^{L}_{O}$ ($U^{R}_{O}$) to represent the left (right) transversal of the automorphism group of the object $O$. A left transversal of a subgroup $H \subset G$ is a list consisting of exactly one element of each left coset of $H$. A right transversal is the analogous object for right cosets.

\subsection{Orbit Stabilizer Theorem}
\label{section:orbit}

The key technical idea from group theory underpinning our techniques is the concept of group actions \citep{Smith2015}.

\begin{definition}
\label{def:groupaction}

\textbf{Group Actions}\\
Given a group $G$ and a set $X$, a group action $\alpha: G \times X \to X$ is a function satisfying the following axioms:

\begin{itemize}

\item If $e$ is the identity of $G$ and $x$ is any element of $X$, $\alpha(e, x) = x$.

\item Given any $g, g^{\prime} \in G$ and $x \in X$, $\alpha\left(g^{\prime},\alpha(g, x)\right) = \alpha\left(g^{\prime}g, x\right)$.

\end{itemize}

\end{definition}

We will write $\alpha\left(g, x\right)$ as $\alpha_g(x)$ going forward.

Ideally, we use group actions to describe the symmetries of a set. A simple example of a group action is left multiplication, $g\cdot x$, where $x$ is a member of the group itself. A more useful example, and one we will indeed use in this paper, is the action of the symmetric group, $S_n$, on the set of simple $n$-vertex graphs. The action of a permutation $\pi \in S_n$ on a graph $H$ is to return a graph $H^{\prime}$ isomorphic to $H$\footnote{We call these graphs different if their adjacency matrices are not equal.} such that the vertices and edges of $H$ are permuted according to $\pi$.

\begin{definition}
\label{def:autgroup}

\textbf{Automorphism Group} or \textbf{Stabilizer}\footnote{In group theoretic literature, this is referred to as the stabilizer, while in combinatorics, especially in graph theory, it is referred to as the automorphism group. Throughout this paper, we will use the term ``automorphism group''.}\\
Given a group action $\alpha:G \times X \to X$, for any $x \in X$, the automorphism group $\textit{Aut}_{G}(x) \subset G$ (or stabilizer $G_x$) is defined as 
$$G_{x} = \textit{Aut}_{G}(x) = \left\{\, g \in G \mid \alpha_g(x) = x \,\right\}.$$

\end{definition}

$\textit{Aut}_{G}(x)$ is to be viewed as the subgroup of elements of $G$ fixing $x$, or equivalently, the subgroup of transformations or actions under which $x$ is invariant.

\begin{definition}
\label{def:orbit}

\textbf{Orbits}\\
Given a group action $\alpha:G \times X \to X$, for any $x \in X$, the orbit of $x$, $\textit{Orb}_{G}(x) \subset X$ is defined as
$$\textit{Orb}_{G}(x) = \left\{\, \alpha_g(x) \mid g \in G \,\right\}.$$

\end{definition}

The orbit of $x$ is to be seen as the set of elements of $X$ that are isomorphic to it. For example, intuitively, the orbit of a graph $H$ is the set of unique labelled graphs isomorphic to $H$. The automorphism group of $H$ is the set of permutations or relabellings that conserve the exact labelled edge relations.

We now state the standard version of the Orbit-Stabilizer Theorem.

\begin{lemma}
\label{lemma:orbit-stabitilizer}

\textbf{Orbit-Stabilizer Theorem (Quantitative Version)}\\
Given a group action $\alpha:G\times X \to X$, for any element $x \in X$, the following relation holds true
$$|\textit{Orb}_{G}(x)||\textit{Aut}_{G}(x)| = |G|.$$

\end{lemma}

More precisely and for better intuition, we state a more qualitative version of this theorem.

\begin{lemma}
\label{lemma:orbit-stabilizer-qualitative}

\textbf{Orbit-Stabilizer Theorem (Qualitative Version)}\\
Given a group action $\alpha:G \times X \to X$, for any element $x \in X$, and any left transversal $U^{L}_{x}$ of the automorphism group $\textit{Aut}_{G}(x)$, we have that $$\textit{Orb}_{G}(x) = \left\{\, \alpha_{u}(x) \mid u \in U^{L}_{x} \,\right\}.$$

\end{lemma}

Hence, we can view members of the orbit of $x$, $\textit{Orb}_{G}(x)$, as ``isomorphic'' to the left cosets of the automorphism group $\textit{Aut}_{G}(x)$, in one sense. The qualitative version of the theorem implies the quantitative version since this is now just a special case of Lagrange's theorem \citep{Smith2015}.

\section{The Unified Meta Theorem}

In this section, we prove our main theorem, showing a recovery reduction for $\SLNP^S$. We organize this section based on the various ingredients in the algorithm.

\subsection{Determining Automorphism Group Size}

It follows from the work of \cite{Sims1970} and \cite{Furst1980} that using the efficiently computable group action $\alpha$ as our efficient membership test, we have a polynomial-time deterministic algorithm to compute the size of the automorphism group $\textit{Aut}_{S_m}(\phi)$ of any instance $\phi$.

\begin{lemma}
\label{lem:autsize}

Suppose we have a function $f$ in the class $\SLNP^S$. Given any instance $\phi \in \Sigma^{p(n)}$, there is a $\poly(n)$-time algorithm to determine the order $\left|\textit{Aut}_{S_m}(\phi)\right|$ of the automorphism group of $\phi$.

\end{lemma}

\begin{proof}

From Definition \ref{def:symgroupsyminv} it follows that the variable $m$ is polynomial in $n$ and from Defintion \ref{def:syminv}, it follows that the group action $\alpha:S_m \times \Sigma^{p(n)} \to \Sigma^{p(n)}$ is computable in $\poly(n)$-time. The membership test to check if $\pi \in S_m$ is contained in $\textit{Aut}_{S_m}(\phi)$, given any instance $\phi \in \Sigma^{p(n)}$ is to verify that $\alpha_{\pi}(\phi) = \phi$, which can be done in $\poly(n)$-time. Hence, Lemma \ref{lemma:SS} implies that the group order $\left|\textit{Aut}_{S_m}(\phi)\right|$ can be deterministically computed in $\poly(n)$-time.

\end{proof}

\subsection{Probabilistic Guarantees for Asymmetric Instances}

Now, we prove that for instances with relatively small automorphism groups, we have a deterministic polynomial time procedure that uses queries to the corrupted truth table $\mathtt{T}^{\prime}_{f}$.

First, we prove a lemma that allows us to make the queries deterministically and obtain good probabilistic guarantees over the randomness of the noise operator $\mathcal{N}_{0.5-\epsilon}$. Suppose our query strategy for each input $\phi$ is deterministic and non-adaptive. Our lemma says that with high probability over the choice of corruptions, for every asymmetric input $\phi$, the queries we make retrieve correct answers at least $\left(0.5+\epsilon/2\right)$-fraction of the time.

\begin{lemma}
\label{lemma:smallaut}

Suppose we have a function $f_n:\Sigma^{p(n)} \to \mathcal{D}$ and subsets $$T_1, T_2, \ldots, T_{|\Sigma|^{p(n)}} \subset \Sigma^{p(n)},$$ each of size at least $16 p(n) \ln(|\Sigma|) / \epsilon^2$. With probability at least $1-1/|\Sigma|^{p(n)}$ over the choice of random corruptions of the operator $\mathcal{N}_{1/2-\epsilon}$, the corrupted truth table $\mathtt{T}^{\prime}_{f} = \mathcal{N}_{1/2-\epsilon}\mathtt{T}_{f}$ has at least $\left(1/2+\varepsilon/2\right)$-fraction of instances left uncorrupted within each subset $T_i$.

\end{lemma}

\begin{proof}

Consider the random variables $\left(X_{\phi}\right)_{\phi \in T}$. The random variable $X_{\phi}$ attains the value of $1$ if the entry of the corrupted truth table $\mathtt{T}^{\prime}_{f} = \mathcal{N}_{1/2-\epsilon}\mathtt{T}_{f}$ corresponding to the input $\phi$ is left uncorrupted and $0$ otherwise. We define the random variable $X_{T} = \sum_{\phi \in T}X_{\phi}$ counting the number of uncorrupted entries within $T$.

We use the Chernoff bound \citep{Mitzenmacher2017, Dubhashi1996} with negative correlation to obtain that
\begin{eqnarray}
\label{eqn:1}
\mathcal{P}\left[X_{T}/|T| \leq 1/2 + \epsilon / 2\right] \leq e^{-\epsilon^2|T|/8}.
\end{eqnarray}
This applies when each random variable $X_i \in \{\, 0,1 \,\}$ in the sum is identically distributed on $[0,1]$, but $$\mathcal{P}\left[X_j = 1 | X_i  = 1\right] \leq \mathcal{P}\left[X_j = 1\right]$$ for all $i \neq j$. 

We have that 
$$\mathcal{P}\left[X_{\phi} = 1 | X_{\phi^{\prime}} = 1\right] = \frac{\delta|\Sigma|^{p(n)}-1}{|\Sigma|^{p(n)}-1} < \delta = \mathcal{P}\left[X_{\phi} = 1\right]$$
for any distinct inputs $\phi$ and $\phi^{\prime}$ from $\Sigma^{p(n)}$. This holds since the sum $\sum_{\phi \in \Sigma^{p(n)}}X_{\phi} = \delta|\Sigma|^{p(n)}$ is a conserved quantity and if the truth table $\mathtt{T}^{\prime}_{f}$ is uncorrupted at the entry $\phi^{\prime}$, only $\delta|\Sigma|^{p(n)}-1$ of the other $|\Sigma|^{p(n)}-1$ entries can be left uncorrupted and their random variables are still identically distributed under this condition.

Hence, given any subset $T \subset |\Sigma|^{p(n)}$ with size $|T| \geq 16 p(n) \ln\left(|\Sigma|\right) / \epsilon^2$, we have, using Equation \ref{eqn:1}, that over the choice of random corruptions of $\mathcal{N}_{1/2 - \epsilon}$, $$\mathcal{P}\left[\text{The fraction of uncorrupted entries in }T\text{ is less than }1/2+\epsilon/2\right] \leq |\Sigma|^{-2p(n)}.$$
Assume we have subsets $T_{1}, T_{2}, \ldots, T_{|\Sigma|^{p(n)}} \subset \Sigma^{p(n)}$, each of size at least $16 p(n) \ln|\Sigma| / \epsilon^2$. The probability that over the choice of random corruptions of the operator $\mathcal{N}_{1/2-\epsilon}$ that at least one subset $T_{i}$ with $i \in \left[|\Sigma|^{p(n)}\right]$ has less than $\left(1/2+\epsilon/2\right)$-fraction correct entries is at most $$|\Sigma|^{p(n)}\cdot|\Sigma|^{-2p(n)} = 1/|\Sigma|^{p(n)},$$
due to the union bound.

\end{proof}

Now, using this lemma, we provide our deterministic querying algorithm, using the partial coset transversal algorithm in Lemma \ref{lemma:SS} to help us list distinct isomorphic instances.

\begin{lemma}
\label{lemma:asymalg}

Suppose we have a function $f$ in the class $\SLNP^S$ and $\epsilon > 0$. Given query access to the corrupted truth table $\mathtt{T}_{f}^{\prime} = \mathcal{N}_{1/2-\epsilon}\mathtt{T}_{f}$, there is a deterministic $\poly\left(n, 1/\epsilon\right)$-time procedure that, with probability at least $1-1/|\Sigma|^{p(n)}$ over the choice of corruptions of $\mathcal{N}_{1/2-\epsilon}$, computes $f_n(\phi)$ correctly for every input $\phi \in \Sigma^{p(n)}$ such that $\left|\textit{Aut}_{S_m}(\phi)\right| \leq {m!}/\lp16 p(n) |\Sigma|/ \epsilon^{2}\rp$.

\end{lemma}

\begin{proof}

Suppose that for each input $\phi \in \Sigma^{p(n)}$, $\mathtt{COSET}_{\phi}(k)$ is the list of $k$ (possibly incomplete) coset representatives of $\textit{Aut}_{S_m}(\phi)$, $\left(g_1, g_2, \ldots, g_{k}\right)$ returned by the deterministic procedure for computing a possibly incomplete list of coset representatives implied by Lemma \ref{lemma:SS}.

For each input $\phi \in \Sigma^{p(n)}$ with automorphism group size \\ $\left|\textit{Aut}_{S_m}(\phi)\right| \leq m! / \left(16 p(n) \ln(|\Sigma|) / \epsilon^2\right)$\footnote{Otherwise, we set $T_{\phi} = \Sigma^{p(n)}$ - this just allows Lemma \ref{lemma:smallaut} to apply in a blackbox fashion. We will not actually use such a set $T_{\phi}$ algorithmically.}, suppose $$\mathtt{COSET}_{\phi}\left(16 p(n) \ln(|\Sigma|) / \epsilon^2\right) = \left(g_1, g_2, \ldots, g_{16 p(n) \ln(|\Sigma|) /\epsilon^2}\right)$$ is the list returned by the deterministic procedure in Lemma \ref{lemma:SS}. Then, we define for each input $\phi \in \Sigma^{p(n)}$, the set $T_{\phi}$ as $$T_{\phi} = \left\{\, \alpha_{g_1}(\phi), \alpha_{g_2}(\phi), \ldots, \alpha_{g_{16 p(n) \ln(|\Sigma|) / \epsilon^2}}(\phi) \,\right\}.$$

Since these are coset representatives of the automorphism group $\textit{Aut}_{S_m}(\phi)$ and the orbit stabilizer theorem (Lemma \ref{lemma:orbit-stabitilizer}) guarantees that we have at least $\left(16 p(n) \ln(|\Sigma|) / \epsilon^2\right)$ distinct orbit members (or coset representatives), it is easy to see that each of the images of group action of $\textit{COSET}_{\phi}\left(16 p(n) \ln(|\Sigma|) / \epsilon^2\right)$ is unique. Hence, each set $T_{\phi}$ is of size exactly $\left(16 p(n) \ln(|\Sigma|) /\epsilon^2\right)$. Due to Lemma \ref{lemma:smallaut}, with probability at least $1-1/|\Sigma|^{p(n)}$ over the choice of corruptions of the operator $\mathcal{N}_{1/2-\epsilon}$, the corrupted truth table $\mathtt{T}^{\prime}_{f} = \mathcal{N}_{1/2-\epsilon}\mathtt{T}_{f}$ has at least a $\left(1/2+\epsilon/2\right)$-fraction of entries left uncorrupted over the set $T_{\phi}$ for every input $\phi \in \Sigma^{p(n)}$.

The algorithm is defined as follows. Suppose we are given $\phi \in \Sigma^{p(n)}$ such that $\left|\textit{Aut}_{S_m}(\phi)\right| \leq m! / \left(16 p(n) \ln(|\Sigma|) / \epsilon^2\right)$. We use the deterministic procedure in Lemma \ref{lemma:SS} to compute \\ $16 p(n) \ln(|\Sigma|) / \epsilon^2$ coset representatives of $\textit{Aut}_{S_m}(\phi)$, where the membership test is to apply $\alpha_{\pi}$ to $\phi$ and check whether the image is $\phi$. Upon computing the list $$\mathtt{COSET}_{\phi}\left(16 p(n) \ln(|\Sigma|) /\epsilon^2\right) = \left(g_1, g_2, \ldots, g_{16 p(n) \ln(|\Sigma|) / \epsilon^2}\right),$$ we compute the set $$T_{\phi} = \left\{\, \alpha_{g_1}(\phi), \alpha_{g_2}(\phi), \ldots, \alpha_{g_{16 p(n) \ln(|\Sigma|) / \epsilon^2}}(\phi) \,\right\}.$$
Then, we query the corrupted truth table $\mathtt{T}_{f}^{\prime}$ in the locations listed in $T_{\phi}$. We take the majority value in $\mathcal{D}$ of all the retrieved values, provided one exists. Under the $\left(1-1/|\Sigma|^{p(n)}\right)$-probability guarantee over the choice of corruptions of the operator $\mathcal{N}_{1/2-\epsilon}$, at least $\left(1/2+\epsilon/2\right)$-fraction of answers are uncorrupted. Hence, under these guarantees, the majority is well-defined and equal to $f(\phi)$ due to Lemma \ref{lemma:isoinv}.

The entire procedure requires $\poly(m)\poly(1/\epsilon)\poly(n)$ deterministic time and $O(p(n)/\epsilon^2)$ queries to the corrupted truth table $\mathtt{T}_{f}^{\prime}$. Due to definition of the class $\SLNP^S$ (Definition \ref{def:symgroupsyminv}), $m = \poly(n)$ and $p(n) = \poly(n)$. Hence, the total deterministic runtime of this procedure is polynomial in $n$ and $1/\epsilon$.

\end{proof}

\subsection{Quick Computability for Symmetric Instances}

Now, we aim to handle the case where the automorphism group of the instance is very large - almost as large as $S_m$ itself on the logarithmic scale. Here, we make the following observation about the structure of the class $\SLNP^S$. This follows from the definition of the class $\SLNP^G$ (Definition \ref{def:syminv}).
\begin{equation}\label{eqn:orbits}
  f_n(\phi) = \bigodot_{x \in \mathcal{C}_n}h_n\left(\phi, x\right) = \bigodot_{i \in [j(n)]}\left(\bigodot_{x \in \mathcal{C}^{\beta}_{n, i}}h_n\left(\phi, x\right)\right).
\end{equation}
Note that, by definition, $j(n)$ is a function that is polynomial in $n$. Now, for convenience of writing, let
\begin{equation}\label{eqn:g}
f_{n, i}(\phi) = \bigodot_{x \in \mathcal{C}^{\beta}_{n, i}}h_n\left(\phi, x\right),
\end{equation}
where $\mathcal{C}^{\beta}_{n, i}$ is the $i^{\text{th}}$ orbit of $\mathcal{C}$ as defined in Definition \ref{def:syminv}.

We want to be able to compute each $f_{n, i}(\phi)$ efficiently by some process, leveraging the fact that the instance $\phi$ is highly symmetric and then compute
\begin{equation}\label{eqn:fintermsofg}
  f_n(\phi) = \bigodot_{i \in [j(n)]}f_{n, i}(\phi)
\end{equation}
in polynomial time since $j(n) = O\left(\poly(n)\right)$.

Indeed, we do compress the brute-force enumeration process by using the symmetries of the instance $\phi$. First, making the following observation.

\begin{observation}
\label{obs:inv}

Given that $h_n\left(\alpha_g(\phi), \beta_g(x)\right) = h_n(\phi, x)$ for every $g \in G_n$, $\phi \in \Sigma^{p(n)}$, $x \in \mathcal{C}_n$, for every element $g$ in the automorphism group $\textit{Aut}_{S_m}(\phi)$, we have that $$h_n\left(\phi, \beta_{g}(x)\right) = h_n(\phi, x).$$

\end{observation}

This follows immediately from the definition of the automorphism group (Definition \ref{def:autgroup}).

Now, we prove our main lemma making our algorithm for high-symmetry instances possible. This lemma, in essence, says that we can reduce the time used to verify all possible proofs by not checking proofs we know the validity of due to symmetry. 

\begin{lemma}
\label{lemma:compression}

\textbf{Compressing Brute-Force via Symmetry}\\
Suppose we have a function $f$ contained in the class $\SLNP^S$. Then, for any $n \in \mathbb{N}$, $i \in [j(n)]$, $g \in S_{m(n)}$, and given any $y_i \in \mathcal{C}_{n, i}^{\beta}$, we have that
   $$f_{n, i}(\phi) = \bigodot_{g \in S_m}h_n\left(\phi, \beta_{g}(y_i)\right)$$ and
   $$f_{n, i}(\phi) = \bigodot_{u \in U^{R}_{\phi}}h_n\left(\phi, \beta_{u}(y_i)\right),$$ where $U^{R}_{\phi}$ is the coset right transversal of the automorphism group $\textit{Aut}_{S_m}(\phi)$.
   
\end{lemma}

\begin{proof}

\begin{enumerate}

\item Now, suppose $\textit{Aut}_{S_m}(x)$ is the automorphism group of any element $x \in \mathcal{C}_n$ and $U^{L}_{x}$ is the left transversal of $\textit{Aut}_{S_m}(x)$. Due to the Orbit-Stabilizer theorem (Lemma \ref{lemma:orbit-stabilizer-qualitative}), we have that for each $x^{\prime} \in \textit{Orb}_{S_m}(x)$, there is a coset representative $u \in U^{L}_x$ such that $x^{\prime} = \beta_{u}(x)$. And hence, may write $h_n\left(\phi, x^{\prime}\right) = h_n\left(\phi, \beta_{u}(x)\right)$. Since $d \odot d = d$ for every $d \in \mathcal{D}$, we may write $h_n\left(\phi, x^{\prime}\right) = h_n\left(\phi, \beta_{u}(x)\right) \odot h_n\left(\phi, \beta_{u}(x)\right) \odot \cdots \odot h_n\left(\phi, \beta_{u}(x)\right)$ any non-zero number of times. Since $\beta_{ug}(x) = \beta_{u}(x)$ for every $g \in \textit{Aut}_{S_m}(x)$, we have that $$h_n\left(\phi, x^{\prime}\right) = \bigodot_{g \in \textit{Aut}_{S_m}(x)}h_n\left(\phi, \beta_{ug}(x)\right).$$
Subsequently, due to the Orbit-Stabilizer theorem (Lemma \ref{lemma:orbit-stabilizer-qualitative}), enumerating over all members $x^{\prime}$ of the orbit of $x$, we have that
\begin{equation*}
\begin{split}
\bigodot_{x^{\prime} \in \textit{Orb}_{S_m}(x)}h_n\left(\phi, x^{\prime}\right) & = \bigodot_{u \in U^{L}_{x}} h_n\left(\phi, \beta_{u}(x)\right) = \bigodot_{u \in U^{L}_{x}} \left(\bigodot_{g \in \textit{Aut}_{S_m}(x)}h_n\left(\phi, \beta_{ug}(x)\right)\right) \\
& = \bigodot_{g \in S_m}h_n\left(\phi, \beta_{g}(x)\right).
\end{split}
\end{equation*}
Applying this to the orbit $\mathcal{C}^{\beta}_{n, i}$ and its member $y_i$, we get that $$f_{n, i}(\phi) = \bigodot_{g \in S_m}h_n\left(\phi, \beta_{g}\left(y_i\right)\right).$$

\item Using the previous part that $f_{n, i}(\phi) = \bigodot_{g \in S_m}h_n\left(\phi, \beta_{g}\left(y_i\right)\right)$, we rewrite this as $$f_{n, i}(\phi) = \bigodot_{g \in S_m}h\left(\phi, \beta_{g}\left(y_i\right)\right) = \bigodot_{u \in U^{R}_{\phi}}\left(\bigodot_{g \in \textit{Aut}_{S_m}(\phi)}h_n\left(\phi, \beta_{gu}(x)\right)\right)$$ using the Orbit-Stabilizer theorem (Lemma \ref{lemma:orbit-stabilizer-qualitative}). Now, keep in mind that $\beta_{gu}(x) = \beta_{g}\left(\beta_{u}(x)\right)$. Using Observation \ref{obs:inv}, we get that for every $g \in \textit{Aut}_{S_m}(\phi)$, $$h_n\left(\phi, \beta_{gu}(x)\right) = h_n\left(\phi, \beta_{g}\left(\beta_{u}(x)\right)\right) = h_n\left(\phi, \beta_{u}(x)\right).$$
      Using the fact that $d \odot d = d$ for every $d \in \mathcal{D}$, we can now show that $$\bigodot_{g \in \textit{Aut}_{S_m}(\phi)}h_n\left(\phi, \beta_{gu}(x)\right) = \bigodot_{g \in \textit{Aut}_{S_m}(\phi)}h_n\left(\phi, \beta_{u}(x)\right) =  h_n\left(\phi, \beta_{u}(x)\right).$$
      Subsequently, we have that $$f_{n, i}(\phi) = \bigodot_{u \in U^{R}_{\phi}}\left(\bigodot_{g \in \textit{Aut}_{S_m}(\phi)}h_n\left(\phi, \beta_{gu}(x)\right)\right) = \bigodot_{u \in U^{R}_{\phi}}h_n\left(\phi, \beta_{u}(y_i)\right).$$
      
\end{enumerate}

\end{proof}

And now, we prove the existence of our polynomial-time algorithm to compute each subfunction $f_{n, i}$.

\begin{lemma}
\label{lemma:f_i}

Suppose we have a function $f$ in the class $\SLNP^S$. Then given any input $\phi \in \Sigma^{p(n)}$ such that the automorphism group size is $|\textit{Aut}_{S_m}(\phi)| \geq m! /\lp16 p(n) \ln(|\Sigma|) /\epsilon^2\rp$ and $i \in [j(n)]$, $f_{n, i}(\phi)$ is deterministically computable in $\poly(n, 1/\epsilon)$-time.

\end{lemma}

\begin{proof}

From the definition of the class $\SLNP^G$ (Definition \ref{def:syminv}), for each $i \in [j(n)]$, we can compute a member $x_i$ of the orbit $\mathcal{C}^{\beta}_{n, i}$ in $\poly(n)$ time. From Definition \ref{def:symgroupsyminv}, we have that the time complexity of computing $x_i$ is bounded by a polynomial in $n$. Using part 2 of Lemma \ref{lemma:compression}, we have that $$f_{n, i}(\phi) = \bigodot_{u \in U^{R}_{\phi}}h_n\left(\phi, \beta_{u}(x_i)\right).$$
We compute the right transversal $U^{R}_{\phi}$ using the deterministic procedure specified in Lemma \ref{lemma:SS} for coset representatives. Since $|\textit{Aut}_{S_m}(\phi)| \geq m! /\lp16 p(n) \ln(|\Sigma|) / \epsilon^2\rp$, due to the Orbit-Stabilizer theorem (Lemma \ref{lemma:orbit-stabilizer-qualitative}), we have that $\left|U^{R}_{\phi}\right| \leq 16 p(n) \ln(|\Sigma|) / \epsilon^2$. Hence, due to Lemma \ref{lemma:SS}, $U^{R}_{\phi}$ is computable in $\poly(n, 1/\epsilon)$-time and contains $\poly(n, 1/\epsilon)$ members.

Now, we enumerate over $U^{R}_{\phi}$ and compute $h_n\left(\phi, \beta_{u}\lp x_i\rp\right)$ for each $u \in U^{R}_{\phi}$. We can compute each of these answers in $\poly(n)$-time due to the definition of the class $\SLNP^S$ (Definition \ref{def:syminv} and Definition \ref{def:symgroupsyminv}). Using all these answers, we can compute $f_{n, i}(\phi) = \bigodot_{u \in U^{R}_{\phi}}h_n\left(\phi, \beta_{u}\lp x_i\rp\right)$. Since the right transversal $U^{R}_{\phi}$ contains polynomially many elements, we only evaluate the function $h_n$ and perform the operation $\odot$ polynomially many times.

\end{proof}

Now, we combine these $f_{n, i}$'s to obtain our value $f_n(\phi)$.

\begin{lemma}
\label{lemma:highlysymalg}

Suppose we have a function $f$ contained in the class $\SLNP^S$ and are given an input $\phi \in \Sigma^{p(n)}$ with the promise that the automorphism group size is $\left|\textit{Aut}_{S_m}(\phi)\right| \geq m! /\lp16 p(n) \ln(|\Sigma|) / \epsilon^2\rp$. Then, $f_n(\phi)$ is deterministically computable in $\poly \left(n, 1/\epsilon\right)$-time.

\end{lemma}

\begin{proof}

Using the $\poly(n, 1/\epsilon)$-algorithm in Lemma \ref{lemma:f_i}, we enumerate $f_{n, i}(\phi)$ for every $i \in [j(n)]$. Following this, we compute $$f_n(\phi) = \bigodot_{i \in [j(n)]}f_{n, i}(\phi).$$
This can be done in $\poly(n, 1/\epsilon)$-time since $j(n) = \poly(n)$ (Definition \ref{def:symgroupsyminv}).

\end{proof}

\subsection{Deterministic Reductions for $\SLNP^S$}

We state our main theorem, giving deterministic polynomial-parameter recovery reductions for our class of functions.

\begin{theorem}

\textbf{(Theorem \ref{theorem:resultsSICSAF} Restated)}\\
Suppose we have a function $f$ in the class $\SLNP^S$. Then, for every $\epsilon > 0$, we have a \\$\left(\poly\left(n, 1/\epsilon\right), 0.5 - \epsilon, \exp\left(-\Omega(\poly(n))\right)\right)$-recovery reduction in the random noise model for $f$.

\end{theorem}

\begin{proof}

First, given the input $\phi \in \Sigma^{p(n)}$, we compute the size of the automorphism group $\textit{Aut}_{S_m}(\phi)$ using the deterministic procedure implied by Lemma \ref{lem:autsize}. This takes $\poly(n)$-time. If $\left|\textit{Aut}_{S_m}(\phi)\right| \geq m! /\lp16 p(n) \ln(|\Sigma|) / \epsilon^2\rp$, we use the deterministic procedure for highly symmetric instances specified in Lemma \ref{lemma:highlysymalg} to compute $f_n(\phi)$ in $\poly(n, 1/\epsilon)$-time. Otherwise, if $\left|\textit{Aut}_{S_m}(\phi)\right| < m! /\lp16 p(n) \ln(|\Sigma|) / \epsilon^2\rp$, then we use the $\poly(n, 1/\epsilon)$-time and $\poly(n, 1/\epsilon)$-query procedure specified in Lemma \ref{lemma:asymalg} that returns the correct answer with probability at least $1-1/|\Sigma|^{p(n)} = 1 -\exp\left(-\Omega(\poly(n))\right)$ over the choice of random corruptions of the operator $\mathcal{N}_{0.5 - \epsilon}$.

\end{proof}

\begin{remark}
\label{remark:multiplication}

Our paradigm so far has been to query on isomorphic instances wherever possible and compute the answer ourselves when we cannot obtain probabilistic guarantees. We note here that this paradigm falls short for some functions and does not always hold as a black-box paradigm. The example is the function $f:\{\, 0,1 \,\}^{n}\times\{\, 0,1 \,\}^n \to \{\, 0,1 \,\}^{2n+1}$ multiplying two $n$ bit integers together. Consider the task of multiplying two $n$-bit prime numbers $p$ and $q$ larger than $2^{n-1}$. If we wanted to use the ``query isomorphic instances'' idea, we can only use the inputs $(p, q)$ and $(q, p)$ since these are the only factorizations of $pq$ with both inputs representable in $n$ bits. The probabilistic guarantees here are poor. If we apply $\mathcal{N}_{0.49}$ to the multiplication table, with probability $1-0.51^2 = 0.7399$, at least one of the truth table entries for the inputs $(p, q)$ and $(q, p)$ is false and looking for a majority is no longer a fruitful approach. Hence, we must manually compute this answer ourselves. Due to the work of \cite{Harvey2021}, there is a $O(n\log n)$ time algorithm for multiplying $2$ $n$ bit integers. Under the \textit{Network Coding Conjecture} \citep{Li2004, Langberg2009}, this algorithm was shown to be optimal \citep{Afshani2019}. If integer multiplication has a non-trivial recovery reduction ($o(n \log n)$-time) under our paradigm, then for every $n$-bit prime pair $p$ and $q$, both larger than $2^{n-1}$, we would be able to multiply them in $o(n \log n)$-time. It seems unlikely to us that multiplying prime numbers should be asymptotically easier than multiplying arbitrary $n$ bit numbers.

\end{remark}

\section{Random Noise Reductions for $\NP$-Hard Problems}
\label{section:randomnp}

We now give the lemma proving that our $\NP$-hard problems are indeed contained in the class $\SLNP^S$. The proof is straightforward, but tedious, so we provide the proof in Appendix \ref{appendix:B}.

\begin{lemma}
\label{lemma:nphsicsaf}

The following functions are contained in the class $\SLNP^S$ (as defined in Definition \ref{def:symgroupsyminv}).

\begin{enumerate}

\item $\SAT$

\item $k\SAT$
  
\item $k\CSP$
  
\item Max-$k\CSP$
  
\item $\INDSET$
  
\item $\VERTEXCOVER$
  
\item $\CLIQUE$
  
\item $k\COL$
  
\item $\HAMCYCLE$
  
\item $\HAMPATH$
  
\end{enumerate}

\end{lemma}

Now, we complete the proof of our main theorem, putting all the ingredients together.

\begin{theorem}

\textbf{(Theorem \ref{theorem:results} Restated)}\\
For every $\epsilon > 0$, the following functions have deterministic\\ $\left(\poly\left(n, 1/\epsilon\right), 0.5 - \epsilon, \exp\left(-\Omega\left(\poly(n)\right)\right)\right)$-recovery reductions in the random noise model:

\begin{enumerate}

\item $\SAT$
  
\item $k\SAT$
  
\item $k\CSP$
  
\item Max-$\CSP$
  
\item $\INDSET$
  
\item $\VERTEXCOVER$
  
\item $\CLIQUE$
  
\item $k\COL$
  
\item $\HAMCYCLE$
  
\item $\HAMPATH$
  
\end{enumerate}

\end{theorem}

\begin{proof}

We note that the semilattices in question are $\left(\{\, 0,1 \,\}, \vee\right)$ and $\left(\mathbb{Z}, \max\right)$ and $b \vee b = b$ for every $b \in \{\, 0,1 \,\}$ and $\max\{\, r, r \,\} = r$ for every $r \in \mathbb{Z}$. Hence, this follows immediately from Lemma \ref{lemma:nphsicsaf} and Theorem \ref{theorem:resultsSICSAF}.

\end{proof}

\section{Random Noise Reductions for Fine-Grained Problems}
\label{section:randomp}

Now, we provide the claimed recovery reductions for fine-grained problems. We start with the $\OV$ problem.

\begin{theorem}

\textbf{(Theorem \ref{thoerem:fg}, Part 1 Restated)}\\
For every $\epsilon = 1/\polylog(n)$ and dimension $d = O\left(n^{1-\gamma}\right)$ for some $\gamma > 0$, there is a\\ $\left(\tilde{O}(nd), 0.5 - \epsilon, 1 - \exp\left(-\Omega(nd)\right) \right)$-recovery reduction in the random noise model for the $\OV$ problem in $d$ dimensions.

\end{theorem}

\begin{proof}

Suppose the input is given as $n$, dimension $d$ vectors $V = \left(v_1, v_2, \ldots, v_n\right)$. First, we remark that we use the group action $\alpha: S_n \times \{\, 0,1 \,\}^{nd} \to \{\, 0,1 \,\}^{nd}$ such that $\alpha_{\pi}(V) = \left(v_{\pi(1)}, v_{\pi(2)}, \ldots, v_{\pi(n)}\right)$. 

 Notice that if $V$ had four or more distinct vectors, then the index of the automorphism group $\textit{Aut}_{S_n}(V)$ is at least $\Omega \left(n^3\right)$. Hence, if the input $V$ has at most three distinct vectors, in $O(nd)$-time, we can scan the input, list the (at most three) distinct vectors and compute $O(1)$ dot products in time $O(nd)$. If one of the possible dot products is $0$, we return $1$ for the $\OV$ problem.

Now, suppose the input $V$ has four or more distinct vectors - we can check this in $O(nd)$-time. Then the automorphism group is of size $O\left(n! / n^3\right) = o\left(n!/n^2\right)$ for all sufficiently large $n$. Now, suppose we have the corrupted truth table $\mathtt{T}^{\prime}_{f} = \mathcal{N}_{0.5 - \epsilon}\mathtt{T}_{f}$. Suppose for each instance $V$ with automorphism group index $\omega\left(n^2\right)$, we define $S_{V} = \textit{Orb}_{S_n}(V)$\footnote{If the index of the automorphism group $\textit{Aut}_{S_n}(V)$ is $O(n^2)$, we set $S_{V} = \{\, 0,1 \,\}^{nd}$ as a formality.}. Now, due to Lemma \ref{lemma:smallaut}, with probability $1-2^{-nd}$, every such set $S_{V}$ with $\textit{Aut}_{S_n}(V)$ index size at least $\omega\left(n^2\right) > 16 nd \ln(2) / \epsilon^2$ has at least $\left(0.5+\epsilon/2\right)$-fraction of its instances correct. Now, our strategy is to sample $O\left (\log n/\epsilon\right)$ random permutations from $S_n$, query the corrupted truth table $\mathtt{T}^{\prime}_{f}$ on the input $\alpha_{\pi}(V)$ for each sampled permutation $\pi$. Due to the Chernoff bound \citep{Mitzenmacher2017}, with probability at least $1- 1/n > 2/3$, the majority is the correct answer. This takes $\tilde{O}(nd/\epsilon^2) = \tilde{O}(nd)$-time. Hence, we have our randomized recovery reduction.

\end{proof}

Now, we provide our recovery reduction for the \textit{Parity $k$-Clique} Problem.

\begin{theorem}

\textbf{(Theorem \ref{thoerem:fg}, Part 2 Restated)}\\
For any constant $k > 0$ and $\epsilon = 1/\polylog(n)$, we have a\\ $\left(\tilde{O}(n^2), 0.5-\epsilon, 1-2^{-{{n} \choose {2}}}\right)$-recovery reduction in the random noise model for \textit{Parity $k$-Clique}, the problem of computing the lowest order bit of the number of $k$ cliques in a graph.

\end{theorem}

\begin{proof}

First, we note that due to Lemma \ref{lemma:23} and Lemma \ref{lemma:24} in Appendix \ref{appendix:C}, we have that for sufficiently large $n$, if a graph has automorphism group size $\left|\textit{Aut}(H) \right| = \omega\left(n! / n^3\right)$, then it is one of twelve graphs, each of which we can recognize and count the number of $k$-cliques on (and hence the parity bit) in $\tilde{O}(n^2)$-time.

If it is not one of those twelve graphs, then $|\textit{Aut}(H)| = O\left(n! / n^3\right)$, and hence, the index of the automorphism group is $\Omega\left(n^3\right) > 16  \ln(2) n^2 / \epsilon^2$. In this case, Lemma \ref{lemma:smallaut} says that with probability at least $1-2^{-{{n} \choose {2}}}$ over the choice of corruptions, the corrupted truth table $\mathtt{T}^{\prime}_{f} = \mathcal{N}_{0.5 - \epsilon}\mathtt{T}_{f}$ has at least $\left(0.5+\epsilon/2\right)$-fraction of entries unflipped over every predefined set $S_{H}$ of size at least $16 \ln (2) n^2 / \epsilon^2$. Hence, for each graph $H$ with $\textit{Aut}(H) = O\left(n! / n^3\right)$, we define $S_{H} = \textit{Orb}_{S_n}(H)$ and $S_{H} = \{\, 0,1 \,\}^{{n} \choose 2}$ for the other twelve cases. Hence, since $\textit{Aut}(H) = O\left(n! / n^3\right)$, we can query the truth table on $O \left(\log n / \epsilon \right)$ random members of $\textit{Orb}_{S_n}(H)$ and return the majority value. We do this by randomly sampling permutations $\pi$ from $S_n$ and querying $\mathtt{T}^{\prime}_{f}$ on $\alpha_{\pi}(H)$ for the $O\left(\log n / \epsilon \right)$ randomly chosen permutations $\pi$. Due to the Chernoff bound \citep{Mitzenmacher2017}, with probability $1-1/n > 0$ over the randomness of our algorithm, we get the correct majority answer for \textit{Parity $k$-Clique} on the input $H$.

\end{proof}

\begin{remark}

Note that in both cases, we use randomness so we can perform our reductions in almost-linear time. This is because the deterministic permutation algorithms of \cite{Sims1970} and \cite{Furst1980} tend to add large polynomial overheads to the algorithm.

\end{remark}

\subsubsection*{Acknowledgments}
We are grateful to Emanuele Viola for useful feedback on presentation and for pointing us to \cite{Akavia2008} and \cite{Akavia2006}.

\bibliographystyle{apalike}

\bibliography{main}

\appendix

\section{Properties of Our Generalized Functions}
\label{appendix:A}

\begin{lemma}
\label{lemma:CSAF_appendix}

\textit{(Lemma \ref{lemma:CSAF} restated)}\\
The complexity class $\NP$ is contained in the function class $\SLNP$.

\end{lemma}

\begin{proof}

Note that $\left(\{\, 0,1 \,\}, \vee\right)$ is a semilattice where $y_1\vee y_2$ is always computable in constant time. From the definition of $\NP$, for a language $L \in NP$, we can write $$f(\phi) = \bigvee_{x \in \mathcal{C}}M(\phi, x),$$ where $f$ is the indicator function for $L$, $\mathcal{C}$ is the set of polynomial-sized certificates, and the verifier $M$ runs in $\poly(n)$ time. The set $\mathcal{C}$ of certificates is of size $2^{|x|} = 2^{\poly(n)}$.

\end{proof}

\begin{lemma}
\label{lemma:isoinv_appendix}

\textbf{Isomorphism Invariance Property of $\SLNP^G$}\textit{ (Lemma \ref{lemma:isoinv} restated)}\\
Given an indicator function $f \in \SLNP^G$, with group sequence $G_1, G_2, \ldots$ and group actions $\alpha$ and $\beta$, we have that for every element $g \in G_n$ and input $\phi \in \Sigma^{p(n)}$, $$f_n\left(\alpha_g(\phi)\right) = f_n(\phi).$$

\end{lemma}

\begin{proof}

First we expand out $f_n\left(\alpha_g(\phi)\right)$ as $$f_n\left(\alpha_g(\phi)\right) = \bigodot_{x \in \mathcal{C}_n}h_n\left(\alpha_g(\phi), x\right) = \bigodot_{x \in \mathcal{C}_n}h_n\left(\alpha_g(\phi), \beta_{g}\left(\beta_{g^{-1}}(x)\right)\right).$$
Since the group action $\beta_{g^{-1}}$ is a permutation of the set $\mathcal{C}_n$ and the operator $\odot$ is commutative, we can rearrange this ``big sum'' as $$f_n\left(\alpha_g(\phi)\right) = \bigodot_{x \in \mathcal{C}_n}h_n\left(\alpha_g(\phi), \beta_{g}(x)\right).$$
Since, pointwise for every $x \in \mathcal{C}_n$, by definition of $\SLNP^G$ (Definition \ref{def:syminv}), we have that $h_n\left(\alpha_{g}(\phi), \beta_{g}(x)\right) = h_n(\phi, x)$, we rewrite this sum as $$f_n\left(\alpha_g(\phi)\right) = \bigodot_{x \in \mathcal{C}_n}h_n\left(\alpha_g(\phi), \beta_{g}(x)\right) = \bigodot_{x \in \mathcal{C}_n}h_n\left(\phi, x\right) = f_n(\phi),$$ and subsequently prove the desired equality.

\end{proof}

\section{Our $\NP$-Hard Problems are in $\SLNP^S$}
\label{appendix:B}

\begin{lemma}

\textit{(Lemma \ref{lemma:nphsicsaf} restated)}\\
The following functions are in the class $\SLNP^S$ as defined in Definition \ref{def:symgroupsyminv}.

\begin{enumerate}

\item $\SAT$
  
\item $k\SAT$
  
\item $k\CSP$
  
\item Max-$k\CSP$
  
\item $\INDSET$
  
\item $\VERTEXCOVER$
  
\item $\CLIQUE$
  
\item $k\COL$
  
\item $\HAMCYCLE$
  
\item $\HAMPATH$
  
\end{enumerate}

\end{lemma}

\begin{proof}

\begin{enumerate}

\item Suppose we are given SAT formulae of length at most $m(n)$ (polynomial in $n$) (in any measure) on $n$ variables. Since $\SAT$ is in $\NP$, due to Lemma \ref{lemma:CSAF}, $\SAT$ is in $\SLNP$. We use the group $S_n$ to act on the input $\phi$ and the assignment $x$ as follows. For any permutation $\pi \in S_n$, $\alpha_{\pi}(\phi)$ relabels every variable $x_i$ as $x_{\pi(i)}$. Given the assignment $x^{\prime} = \left(x^{\prime}_{1}, x^{\prime}_{2}, \ldots, x^{\prime}_{n}\right) \in \{\, 0,1 \,\}^{n}$, the permutation $\pi \in S_n$ acts on $x^{\prime}$ as $\beta_{\pi}\left(x^{\prime}_1, x^{\prime}_2, \ldots, x^{\prime}_n\right) = \left(x^{\prime}_{\pi^{-1}(1)}, x^{\prime}_{\pi^{-1}(2)}, \ldots, x^{\prime}_{\pi^{-1}(n)}\right)$ - the entry in the $j^{\text{th}}$ position moves to the $\pi(j)^{\text{th}}$ position. In the view that $f(\phi) = \bigodot_{x \in \{\, 0,1 \,\}^n}M(\phi, x)$, it is easy to see that for every $\pi \in S_n$, $$M\left(\alpha_{\pi}(\phi), \beta_{\pi}(x)\right) = M(\phi, x),$$ and that both group actions $\alpha$ and $\beta$ are computable in $\poly(n)$-time. The challenge remains to show that $\{\, 0,1 \,\}^n$ partitions into $\poly(n)$ orbits under the action of $\beta$ and that we can deterministically sample a list of elements of $\{\, 0,1 \,\}^n$ in each orbit in $\poly(n)$-time. This is, indeed, the case. The group action $\beta$ partitions $\{\, 0,1 \,\}^n$ into $n+1$ partitions on the basis of the number of $0$s in the string. A representative of the orbit with $k$ $0$s is $0^k1^{n-k}$. Hence, with all relevant polynomial parameters, $\SAT$ is in $\SLNP^S$.
  
\item This proof follows almost immediately from the proof for $\SAT$. The only difference here is that the input $\phi$ has clause-width at most $k$. We can also represent $\phi$ as a string in $\{\, 0,1 \,\}^{\sum_{j \in [k]}{{2n}\choose{k}}}$ choosing or not choosing one of the $\sum_{j \in [k]}{{2n} \choose {j}}$ unrestricted clauses of length at most $k$.

\item For $k\CSP$, we have a clause $C:\Sigma^{k} \to \{\, 0,1 \,\}$ with an arbitrary truth table. One can see that any $C\left(x_1, \ldots, x_{k}\right)$ is computable in $O(1)$ time. We construct our function $h_n:\{\, 0,1 \,\}^{\left(q(n)\right)^k} \times \mathcal{C}_n \to \{\, 0,1 \,\}$ where $\mathcal{C}_n = \Sigma^{q(n)}$ and $q$ is a polynomial representing the length of the certificate. Here, we have that 
$$h_n(\phi, y) = \bigwedge_{(i_1, i_2, \ldots, i_k) \in [q(n)]^k} C_{\text{select}}\left(\phi, y_{i_1}, y_{i_2}, \ldots, y_{i_k}\right),$$ where 
$$C_{\text{select}}(\phi, y_1, y_2, \ldots, y_k) = 
\begin{cases}
        C\left(y_1, y_2, \ldots, y_k \right), & \mbox{if } \phi \text{ selects the tuple }\lp i_1, i_2, \ldots, i_k \rp; \\
        1, & \mbox{otherwise}.
      \end{cases}$$
      
Note that $h$ can be computed in polynomial time since $q$ is a polynomial in $n$. We can see $\phi$ as a $\left(q(n)\right)^k$-bits long string selecting and deselecting clauses in the $k\CSP$ instance. Our string $y$ is a proof from the set $\Sigma^{q(n)}$. Our $k\CSP$ can be represented as $$f_n(\phi) = \bigvee_{y \in \Sigma^{q(n)}}h_n(\phi, y).$$
      
Now, we describe our group actions $\alpha: S_{q(n)} \times \{\, 0,1 \,\}^{\left(q(n)\right)^{k}} \to \{\, 0,1 \,\}^{\left(q(n)\right)^{k}}$ and $\beta: S_{q(n)} \times \Sigma^{q(n)} \to \Sigma^{q(n)}$. The simpler group action $\beta$ acts on $y = \left(y_1, y_2, \ldots, y_{q(n)}\right)$ to give $\beta_{\pi}(y) = \left(y_{\pi(1)}, y_{\pi(2)}, \ldots, y_{\pi(q(n))}\right)$ for every permutation $\pi \in S_{q(n)}$. The group action $\alpha$ acts to preserve the relation, with the bit $\alpha_{\pi}(\phi)_{\left(i_1, i_2, \ldots, i_k\right)} = \phi_{\left(\pi^{-1}\left(i_1\right), \pi^{-1}\left(i_2\right), \ldots, \pi^{-1}(i_k)\right)}$ for every permutation $\pi \in S_{q(n)}$. It can be seen that the relation $h\left(\alpha_{\pi}(\phi), \beta_{\pi}(y)\right) = h(\phi, y)$ is preserved for every $\pi \in S_{q(n)}$, instance $\phi \in \{\, 0,1 \,\}^{\left(q(n)\right)^k}$, and certificate $y \in \Sigma^{q(n)}$. Note that $q$ is a polynomial and $S_{q(n)}$ is the group and hence, this is in $\SLNP^S$ for every possible clause $C:\Sigma^{k} \to \{\, 0,1 \,\}$.
      
\item The proof proceeds identically to the above case, except we use $\left(\mathbb{Z}, \max\right)$ as the semilattice.
  
\end{enumerate}

For cases 5 through 10, due to Lemma \ref{lemma:CSAF}, it is easy to see that these are in $\SLNP$. Now, we prove that they are in $\SLNP^S$ due to action from $S_n$ on $n$ vertex graph inputs. The group action $\alpha$ always acts on a graph $H$ by permuting the graph according to permutation $\pi$, where the graph $\alpha_{\pi}(H)$ contains an edge between $v_{\pi(i)}$ and $v_{\pi(j)}$ if and only if the graph $H$ has an edge between $v_i$ and $v_j$. Now, we prove that for each of these problems, we have a group action $\beta$ preserving symmetry and splitting the certificate space $\mathcal{C}$ into $\poly(n)$-partitions, each of which we can deterministically produce one canonical certificate from. Since one can verify the parameters, this is sufficient to show that these functions are in $\SLNP^S$.

\begin{enumerate}

\setcounter{enumi}{4}

\item Here, the certificate is a set $S \subset [n]$. The group action $\beta_{\pi}$ acts on $S$ by giving us a set $\beta_{\pi}(S)$ that contains $\pi(i)$ if and only if $S$ contains $i$. One can see that the set $\beta_{\pi}(S)$ is an independent set for the graph $\alpha_{\pi}(H)$ if and only if the set $S$ is an independent set for the graph $H$. Here, the orbits partition the power set of $[n]$, $\mathcal{P}\left([n]\right)$ into $n+1$ partitions on the basis of size. We can easily sample the set $\{\, 1, 2, \ldots, k \,\}$ as the canonical certificate of size $k$. Here, our function $h$, on input $k$, graph $H$ and set $S$ checks if $S$ encodes an independent set for the graph $H$ and that $|S| = k$.
  
\item Here, this follows identically to the above case except the function $h$ checks if the set $S$ is a vertex cover of size $k$ for the graph $H$.
  
\item As above, the proof proceeds identically, except the function $h$ checks if the set $S$ forms a clique of size $k$ in the graph $H$.
  
\item Here, our certificate is coloring $C = \left(c_1, c_2, \ldots, c_n\right) \in [k]^{n}$ of the vertices. The function $h$ checks if the coloring $C$ is a valid $k$-coloring of the graph $H$, with no monochromatic edges. The group action $\beta$ acts on the coloring $C$ by producing $\beta_{\pi}(C) = \left(c_{\pi^{-1}(1)}, c_{\pi^{-1}(2)}, \ldots, c_{\pi^{-1}(n)}\right)$. One can see that $\beta_{\pi}(C)$ is a valid coloring of the graph $\alpha_{\pi}(H)$ is and only if $C$ is a valid coloring for the graph $H$. It can also be seen that there are ${{n + k - 1}\choose {k - 1}} = O\left(n^{k-1}\right)$ orbit partitions on the basis of the number of each color $c \in [k]$ present in the coloring $C$ - we can sample one canonical member of each orbit in polynomial-time.
  
\item Here, the certificate is a sequence $I = (i_1, i_2, \ldots, i_n)$ where each distinct $i_k \in [n]$ and this represents a cycle $i_1 \to i_2 \to \cdots \to i_n \to i_1$. The function $h$ checks if the sequence is a valid Hamiltonian cycle for the graph $H$. We note the redundancy but introduce it for the sake of symmetry. The action $\beta_{\pi}$ acts on the sequence $I$ by returning $\beta_{\pi}(I) = \left(i_{\pi^{-1}(1)}, i_{\pi^{-1}(2)}, \ldots, i_{\pi^{-1}(n)}\right)$. One can see that there is only one orbit and that the relation is preserved upon group action.
  
\item Here, the certificate is once again a sequence $I = (i_1, i_2, \ldots, i_n)$ as before, representing the path $i_1 \to i_2 \to \cdots \to i_n$. Once again, we note that we introduce the redundancy in certification to allow for symmetry. The function $h$ checks if the path encoded by the sequence $I$ is a Hamiltonian path in the graph $H$. We use the group action $\beta$ identically to before and there is only one orbit, with the group actions preserving the relation.
  
\end{enumerate}

\end{proof}

\section{A Classification of Graphs With $\left| \textit{Aut} \left( H \right) \right| = \omega \left( n! / n^3 \right)$}
\label{appendix:C}

First, we will prove Lemma \ref{lemma:22} that gives the properties of the $n$ vertex graphs with $\left| \textit{Aut} \left( H \right) \right| = \omega \left( n! / n^3 \right)$ in terms of the number of partitions based on the degree of vertices, the size of such partitions, and the degree distribution of the vertices.

\begin{lemma}
\label{lemma:22}

For sufficiently large $n$, any graph $H$ with $\left| \textit{Aut} \left( H\right) \right| = \omega \left( n! / n^3 \right)$ satisfies the following properties.

\begin{enumerate}

\item If the vertices of $H$ are partitioned based on degree, then there are at most three partitions.

\item No partition can be simultaneously larger than $2$ and smaller than $n-2$.

\item The degree of any vertex $v$ can be in the set $\{\, 0, 1, 2, n-2, n-1 \,\}$.

\end{enumerate}

\end{lemma}

\begin{proof}

Suppose that we have $m \geq 4$ partitions of size $(\alpha_i)_{i \in [m]}$, in ascending order, with each $\alpha_i \geq 1$. The probability that $\pi \in S_n$ is in $\textit{Aut} \left( H \right)$ is bounded from above by
\begin{equation*}
\frac{\prod_{i \in [m]} \alpha_i!}{n!} \leq \frac{(\sum_{i \in [m-1]} \alpha_i)! \alpha_m!}{n!},
\end{equation*}
since $\pi$ is not allowed to permute vertices across partitions. Since $\alpha_m \geq n / m$ due to the pigeonhole principle, and $\alpha_m \leq n - (m - 1)$ due to each $\alpha_i$ being positive, we have that
\begin{equation*}
\frac{\left( \sum_{i \in [m-1]} \alpha_i \right)! \alpha_m!}{n!} = \frac{1}{\displaystyle {n \choose \alpha_m}} \leq \frac{1}{\displaystyle {n \choose 3}} = O \left( \frac{1}{n^3} \right).
\end{equation*}
If this is the case, then $\left| \textit{Aut} \left( H \right) \right|$ is upper bounded by $O\left( n! / n^3 \right)$, leading to a contradiction. This proves the first statement of the lemma.

If there is a partition of size $\alpha$, then the probability that $\pi \in S_n$ is in $\textit{Aut} \left( H \right)$ is upper bounded by $1 / \displaystyle {n \choose \alpha} = O \left( 1 / n^3 \right)$ for the forbidden range. This implies the second statement of the lemma.

Let us consider that the degree $m$ of $v$ in $H$ is greater than $2$ and less than $n-2$. Let the neighbors of $v$ be $(u_i)_{i \in [m]}$. The probability that $\pi \in S_n$ is in $\textit{Aut} \left( H \right)$ is bounded from above by
\begin{equation*}
\frac{n (m!) (n-1-m)!}{n!} = \frac{1}{\displaystyle {{n-1}\choose m}} \leq \frac{1}{\displaystyle {{n -1} \choose 3}} = O(1/n^3),
\end{equation*}
since $n$ is the maximum number of vertices $v$ could map to, $m!$ is the number of ways the neighbors of $v$ could distribute themselves among the neighbors of the image of $v$, and $(n-m-1)!$ is the number of ways the remaining vertices can distribute. Also, $m$ is between $3$ and $n-3$. Due to a similar argument as before, this implies the third statement of the lemma.

\end{proof}

Now, using Lemma \ref{lemma:22}, we will prove Lemma \ref{lemma:23} that gives the structure of the graphs with $\left| \textit{Aut} \left( H \right) \right| = \omega \left( n! / n^3 \right)$.

\begin{lemma}
\label{lemma:23}

Only the following graphs have $\left| \text{Aut} \left( H \right) \right| = \omega \left( n! / n^3 \right)$.

\begin{enumerate}

\item $K_n$ and its complement.

\item $K_n$ with one edge missing and its complement.

\item $K_{n - 1}$ with an isolated vertex and its complement.

\item $K_{n - 1}$ with one vertex of degree $1$ adjacent to it and its complement.

\item $K_{n - 2}$ with two isolated vertices and its complement.

\item $K_{n - 2}$ with two vertices of degree $1$ adjacent to each other and its complement.

\end{enumerate}

\end{lemma}

\begin{proof}

Using Lemma \ref{lemma:23}, the only possible partition sizes based on degree we can have are $(n)$, $(n-1, 1)$, $(n-2, 1, 1)$ and $(n-2, 2)$. Now, by a case-by-case analysis, we will determine which graphs can have such large automorphism groups. Since $\textit{Aut} \left( H \right) = \textit{Aut} \left( \overline{H} \right)$, we will categorize by the degree of the largest partition and assume that the degree is less than or equal to $2$. This way, we will either allow a graph and its complement or reject both. We will also assume that $n$ is sufficiently large, say $n \geq 100$.

\textbf{\textit{Case 1: The Largest Partition Degree is $0$.}} \\
Now, for the $(n)$ partition, the graph is either empty or the complete graph $K_n$. Clearly,
\begin{equation*}
\left| \textit{Aut} \left( H \right) \right| = n! = \omega \left( \frac{n!}{n^3} \right),
\end{equation*}
in both cases, so we allow both.

When the partition is $(n-1, 1)$, this is technically not allowed since even the vertex of the partition of size $1$ must have degree zero, meaning such a partition with these degrees cannot exist.

When the partition is $(n-2, 1, 1)$, this cannot exist since the partitions of size $1$ must have the same degree.

When the partition is $(n-2, 2)$, the only allowed case is that both the vertices in the partition of size $2$ are adjacent. Otherwise, they would also have degree $0$, and we would have $(n)$ again. The other case is $K_n$ with one edge missing. Both of them have
\begin{equation*}
\left| \textit{Aut} \left( H \right) \right| = 2 (n-2)! = \frac{n!}{O(n^2)} = \omega \left( \frac{n!}{n^3} \right),
\end{equation*}
hence, we allow them both.

From this case, we allow the graphs as described in statements 1 and 2 of the lemma.

\textbf{\textit{Case 2: The Largest Partition Degree is $1$.}} \\
For the partition type $(n)$, this is only allowed when $n$ is even due to the handshake lemma. When so, the vertices arrange themselves in pairs. Visually, we have $n / 2$ ``sticks''. We can permute these sticks in $(n / 2)!$ ways and flip them in $2^{n / 2}$ ways. In particular, the size of the automorphism group is
\begin{equation*}
\left| \textit{Aut} \left( H \right) \right| = \left( \frac{n}{2} \right)! \cdot 2^{n / 2} \leq \frac{n!}{n^3} = O \left( \frac{n!}{n^3} \right),
\end{equation*}
for sufficiently large $n$. Hence, we reject this case.

For the partition type $(n-1, 1)$, we have the following possibilities: The vertex in the partition of size $1$ may have possible degrees $n-1$, $n-2$, $2$, or $0$.

\begin{enumerate}

\item The vertices in the $n - 1$-partition are all adjacent to the vertex in the $1$-partition. This is allowed, with
\begin{equation*}
\left| \textit{Aut} \left( H \right) \right| = (n - 1)! = \frac{n!}{n} = \omega \left( \frac{n!}{n^3} \right),
\end{equation*}
and hence, we allow $K_{n-1}$ with an isolated vertex and its complement graph. This covers the case where the $1$-partition vertex has degree $n-1$.

\item If the degree of the $1$-partition vertex is $n-2$, this is disallowed for the following reason: The vertex in the $n-1$-partition not adjacent to the $1$-partition vertex must be adjacent to one of the other vertices, if it needs a degree of $1$. This creates a vertex of degree $2$ in the $n-1$-partition.

\item The degree of the $1$-partition vertex is $2$. In this case, we have two vertices $u$ and $v$ in the $n - 1$-partition that are adjacent to the $1$-partition vertex. The others are arranged similarly to the $(n)$ case for degree $1$. Here, the automorphism group size is
\begin{equation*}
\left| \textit{Aut} \left( H \right) \right| = 2 \left( \frac{n - 3}{2} \right)! \cdot 2^{(n - 3) / 2} = O \left( \frac{n!}{n^3} \right),
\end{equation*}
which for sufficiently large $n$ is too small; hence we reject this case when $n$ is odd. The graph is not possible when $n$ is even

\item If the degree of the $1$-partition vertex is $0$ and the others have degree $1$, this suffers from the same pitfalls as the $(n)$ case, having an automorphism group of size
\begin{equation*}
\left| \textit{Aut} \left( H \right) \right| = \left( \frac{n - 1}{2} \right)! \cdot 2^{(n - 1) / 2} = O \left( \frac{n!}{n^3} \right),
\end{equation*}
and hence, we reject this case as well when $n$ is odd. The graph is not possible when $n$ is even.

\end{enumerate}

For the partition type $(n - 2, 1, 1)$, let the two $1$-partition vertices be $u_1$ and $u_2$ with degrees $d_1$ and $d_2$, respectively. Without loss of generality, assume that $d_2 > d_1$. Since the unique degrees $d_1$ and $d_2$ are different, the $n - 2$-partition implicitly partitions itself into three parts: The partition that is adjacent to the vertex $u_1$ of size $\alpha_1$, the partition that is adjacent to the vertex $u_2$ of size $\alpha_2$, and the remaining vertices that pair themselves. These partitions are rigid in that no $\pi$ from the automorphism group can map vertices across the partition. Hence, assuming the correct parity for $n$, the probability that a random $\pi$ from $S_n$ is in the automorphism group is
\begin{equation*}
\begin{split}
\frac{\left| \textit{Aut} \left( H \right) \right|}{n!} & \leq \frac{\alpha_1! \alpha_2! \displaystyle \left( \frac{n - \alpha_1 - \alpha_2 - 2}{2} \right)! 2^{\left( n - \alpha_1 - \alpha_2 - 2 \right) / 2}}{n!} \\
& \leq \frac{d_1! d_2! \displaystyle \left( \frac{n - \alpha_1 - \alpha_2 - 2}{2} \right)! 2^{\left( n - \alpha_1 - \alpha_2 - 2 \right) / 2}}{n!} \\
& = O \left( \frac{1}{n^3} \right),
\end{split}
\end{equation*}
if $d_1$ and $d_2$ are both from the set $\{\, 0, 2 \,\}$. Therefore, $d_2$ is either $n - 1$ or $n - 2$. The value of $d_2$ cannot be $n - 1$, since then $u_2$ is connected to all the other vertices, forcing $d_1 = 1$, which is not allowed. The only possibility that remains is $d_2 = n - 2$. If $d_1 = 0$, then the vertex $u_1$ is isolated, and $u_2$ is connected to all vertices in the $n - 2$-partition. In this case, we have
\begin{equation*}
\left| \textit{Aut} \left( H \right) \right| = (n - 2)! = \frac{n!}{O \left( n^2 \right)} = \omega \left( \frac{n!}{n^3} \right),
\end{equation*}
so that we allow this graph. We also allow the complement of this graph, a $K_{n-1}$ with a vertex of degree $1$ adjacent to it.

If $d_2 = n - 2$ and $d_1 = 2$, then the vertex $u_2$ is connected to all but one vertex in the $n - 2$-partition, and it is also connected with the vertex $u_1$. The vertex $u_1$ is also connected with the isolated vertex in the $n - 2$-partition. In this case, we have
\begin{equation*}
\left| \textit{Aut} \left( H \right) \right| = (n - 3)! = \frac{n!}{O \left( n^3 \right)} = O \left( \frac{n!}{n^3} \right),
\end{equation*}
so that this graph is rejected.

For the $(n - 2, 2)$ case, we have two vertices $v_1$, and $v_2$ of degree $d$. We have the following cases.

\begin{enumerate}

\item If $d = 0$, we have a case similar to that of $(n)$ with degree $1$, where the automorphism group size is
\begin{equation*}
\left| \textit{Aut} \left( H \right) \right| = 2 \cdot 2^{(n - 2) / 2} \left( \frac{n - 2}{2} \right)! = O \left( \frac{n!}{n^3} \right).
\end{equation*}
We disallow this case when $n$ is even. When $n$ is odd, the graph is not possible.

\item If $d = 1$, this is not allowed since we have defined the partition class this way.

\item For $d = 2$, this implicitly partitions the $n - 2$-partition into two parts: Adjacent to a vertex of degree $2$ and not adjacent to a vertex of degree $2$. Suppose that these vertices are partitioned into partitions of size $\alpha_1$ and $\alpha_2$, respectively, the probability that $\pi \in S_n$ is in the automorphism group is
\begin{equation*}
\frac{\left| \textit{Aut} \left( H \right) \right|}{n!} \leq \frac{2 \cdot \alpha_1! \alpha_2!}{n!} = \frac{2}{n(n-1)} \cdot \frac{1}{\displaystyle {{n-2} \choose \alpha_1}} = O \left( \frac{1}{n^3} \right),
\end{equation*}
since $\alpha_1$ and $\alpha_2$ are necessarily positive. If they were not, we would either have $v_1$ and $v_2$ have very high degree, or degree $0$ or $1$. We reject this graph.

\item For $d = n - 2$ and $d = n - 1$, this is not allowed since at least one vertex from the $n - 2$-partition would have to have a degree larger than $1$.

\end{enumerate}

This case covers the statements 3 and 4 of the lemma.

\textbf{\textit{Case 3: The Largest Partition Degree is $2$.}} \\
For the $(n)$ case, for sufficiently large $n$, we must have $v_1, v_2, v_3, v_4, v_5$, and $v_6$ such that $v_1$ and $v_2$ are adjacent, $v_2$ and $v_3$ are adjacent, $v_4$ and $v_5$ are adjacent, and $v_5$ and $v_6$ are adjacent. If we pick a random permutation $\pi$ from $S_n$, the probability that it is in the automorphism group is
\begin{equation*}
\frac{\left| \textit{Aut} \left( H \right) \right|}{n!} \leq \frac{n \cdot 2 \cdot (n - 3) \cdot 2 \cdot (n - 6)!}{n!} = O \left( \frac{1}{n^4} \right),
\end{equation*}
since $v_2$ can map to at most $n$ vertices, $v_1$ and $v_3$ can only swap their positions as a neighbor of $v_2$; similarly, $v_5$ can map to at most $n - 3$ vertices, $v_4$ and $v_6$ can only swap their positions as a neighbor of $v_5$, and the remaining vertices can map freely to give an upper bound. Hence, we reject this case.

When we have the partition type $(n - 1, 1)$, we have the following cases, based on the degree $d$ of the $1$-partition vertex $u$.

\begin{enumerate}

\item If $d = 0$, this graph suffers the same pitfalls as the $(n)$ partition case and has the automorphism group size of
\begin{equation*}
\left| \textit{Aut} \left( H \right) \right| \leq \frac{(n - 1) \cdot 2 \cdot (n - 4) \cdot 2 \cdot (n - 7)!}{n!} = O \left( \frac{1}{n^5} \right),
\end{equation*}
hence, we reject this case.

\item If $d = 1$, suppose $v_1$ is in the $(n-1)$-partition and adjacent to $u$. If $v_2$ is adjacent to $v_1$, we must find a $v_3$ adjacent to $v_2$ since $v_3$ cannot be adjacent to any of the vertices we already numbered. Otherwise, $u$'s degree would be too high, and a similar case would go for $v_1$ and $v_2$. Once we continue this process and reach $v_{n-1}$, this vertex has no chance of having a degree $2$ since all other vertices have their promised degrees. Such a graph does not exist.

\item If $d = 2$, we violate the definition of our partition structure.

\item If $d = n - 2$, we only have one possibility: Suppose $u$ is adjacent to $v_1$ through $v_{n - 2}$. The vertex $v_{n-1}$ is adjacent to $v_1$ and $v_2$. From $i = 1$ onwards, $v_{2i+1}$ is also adjacent to $v_{2i+2}$. The automorphism group of this graph is of the size of
\begin{equation*}
\left| \textit{Aut} \left( H \right) \right| = 2 \cdot 2^{(n - 4) / 2} \left( \frac{n - 4}{2} \right)! = O \left( \frac{n!}{n^3} \right),
\end{equation*}
since the vertices $v_1$ and $v_2$ can swap themselves; and all other remaining $(n - 4) / 2$ pairs can swap and rearrange themselves. We reject this case when $n$ is even. When $n$ is odd, the graph is not possible. 

\item If $d = n - 1$, then the structure would be $u$, connected to each $v_i$ and the $v_i$'s forming pairs again, like the sticks. The automorphism group is of the size of
\begin{equation*}
\left| \textit{Aut} \left( H \right) \right| = 2^{(n - 1) / 2} \left( \frac{n - 1}{2} \right)! = O \left( \frac{n!}{n^3} \right),
\end{equation*}
for sufficiently large $n$, and we reject this case when $n$ is odd. The graph is not possible when $n$ is even.

\end{enumerate}

When we have a partition structure $(n-2, 2)$, we have the following cases, where $d$ is the degree of the $2$-partition.

\begin{enumerate}

\item If $d = 0$, then this suffers from the same asymptotic pitfalls as the $(n)$-case for degree $2$ and we reject this case:
\begin{equation*}
\left| \textit{Aut} \left( H \right) \right| \leq \frac{2 \cdot (n - 2) \cdot 2 \cdot (n - 5) \cdot 2 \cdot (n - 8)!}{n!} = O \left( \frac{1}{n^6} \right).
\end{equation*}

\item If $d = 1$, suppose $u_1$ and $u_2$ are from the $2$-partition. If $u_1$ and $u_2$ are adjacent, this suffers from the same pitfall as the $(n)$-case again (as shown in case 1 above), and we reject this case. If they are not adjacent, then suppose that $v_1$ is adjacent to $u_1$. The vertex $v_1$ is adjacent to $v_2$. The vertex $v_2$ cannot be adjacent to any of the vertices we visited, so we require a new vertex $v_3$. Similarly, we go on until $v_{n - 2}$. The vertex $v_{n - 2}$ must be adjacent to $u_2$, since all the others already have the promised degree. This resulting graph has an automorphism group size of $2$: Only reflectional symmetry. Another alternative is one chain from $u_1$ to $u_2$, and a cover of cycles. Once again, the $u_1$, $u_2$ component with the chain only has reflectional symmetry, so we have an automorphism group of size
\begin{equation*}
\left| \textit{Aut} \left( H \right) \right| \leq 2 \cdot (n-3)! = \frac{n!}{O \left( n^3 \right)} = O \left( \frac{n!}{n^3} \right),
\end{equation*}
and we reject this case.

\item The case of $d = 2$ is again not allowed.

\item If $d = n - 2$, we have two cases:

\begin{itemize}

\item If $u_1$ and $u_2$ are not adjacent, then they are connected to each vertex of the $n - 2$-partition. This graph has an automorphism group of size
\begin{equation*}
\left| \textit{Aut} \left( H \right) \right| = 2 \cdot (n - 2)! = \frac{n!}{O \left( n^2 \right)} = \omega \left( \frac{n!}{n^3} \right),
\end{equation*}
and we accept this and its complement: $K_{n-2}$ with the other component being an edge.

\item If $u_1$ and $u_2$ are adjacent, then $u_1$ and $u_2$ are adjacent to $n-3$ vertices each in the $n - 2$-partition. The vertices $v_3$ through $v_{n - 2}$ are adjacent to both, and $v_1$ (adjacent to $u_1$) is adjacent to $v_2$ (adjacent to $u_2$). The automorphism group size is
\begin{equation*}
\left| \textit{Aut} \left( H \right) \right| = 2 \cdot (n - 4)! = O \left( \frac{n!}{n^4} \right),
\end{equation*}
and hence, we reject it.

\end{itemize}

\item If $d = n - 1$, then this graph is a complement of $K_{n - 2}$ along with two isolated vertices. This graph has an automorphism group size of
\begin{equation*}
\left| \textit{Aut} \left( H \right) \right| = 2 \cdot (n - 2)! = \frac{n!}{O \left( n^2 \right)} = \omega \left( \frac{n!}{n^3} \right), 
\end{equation*}
and we accept this and its complement.

\end{enumerate}

When we have a partition structure of $(n - 2, 1, 1)$, we reject. We can categorize this graph as follows: The vertices in the $n - 2$-partition form a cycle within the partition, there is a chain starting at $u_1$ and ending at $u_2$, or starting at $u_i$ and ending at $u_i$ (for $i = 1$ or $2$). There must be at least one such cycle containing some $u_i$, since otherwise, $d_1$ would be equal to $d_2$. Let $a$ be the length of the chain and $b$ be the number of such isomorphic chains. The number of permutations in the automorphism group is
\begin{equation*}
\left| \textit{Aut} \left( H \right) \right| \leq 2^a a! ( n - a b)! = O \left( \frac{n!}{n^3} \right),
\end{equation*}
since $b$ is at least $1$ and $a$ is at least $3$.

This case covers the statements 5 and 6 of the lemma.

\end{proof}

Due to the above case-by-case analysis, we have the following lemma.

\begin{lemma}
\label{lemma:24}

For sufficiently large $n$, in $\tilde{O} \left( n^2 \right)$-time, given an $n$-vertex undirected simple graph $H$, we can check whether $\text{Aut} \left( H \right) = \omega \left( n! / n^3 \right)$ and also compute the number of $k$-cliques for any $k > 2$.

\end{lemma}

\begin{proof}

Our algorithm will proceed as follows. If the number of edges in $H$ is larger than $\displaystyle {n \choose 2} / 2$, then we check if it is one of the large-clique structures. If not, then we compute $\overline{H}$ and check for one of the large clique structures. Both counting edges and computing the complement of $H$ requires $O \left( n^2 \right)$-time. Hence, assuming $U^\prime_n = H$ or $\overline{H}$ with at least $\displaystyle {n \choose 2} / 2$ edges, our algorithm proceeds as follows.

\begin{enumerate}

\item \textit{Checking if $U^\prime_n$ is $K_n$:} Simply check if every entry in $U^\prime_n$ is $1$ confirms this. If this test is passed, if $U^\prime_{n} = U_{n}$, then the number of $k$-cliques is $\displaystyle {n \choose k}$. If $U^\prime_{n} = \overline{H}$, then the number of $k$-cliques is $0$. If $U^\prime_{n}$ does not pass this test, then we move to the next test.

\item \textit{Checking if $U^\prime_n$ is $K_n$ with one missing edge:} Simply checking if exactly one entry in $U^\prime_n$ is $0$ confirms this. If the test is passed and $U^\prime_n = H$, then the number of $k$-cliques is $\displaystyle {n \choose k} - {{n - 2} \choose {k - 2}}$, since the subtracted number is the number of $k$-cliques that, in $K_n$, would contain the excluded edge. If $U^\prime_n = \overline{H}$, then the number of $k$-cliques is $0$. If this test fails, then we move to the next test.

\item \textit{Checking if $U^\prime_n$ is $K_{n - 1}$ with an isolated vertex:} It suffices to check if $n - 1$ vertices have degree $n - 2$ and one has degree $0$. If $U^\prime_n$ passes this test and $U^\prime_n = H$, then the number of $k$-cliques is $\displaystyle {{n - 1} \choose k}$. If $U^\prime_n = \overline{H}$, then the number of $k$-cliques is $0$. If this test fails, then we move to the next test.

\item \textit{Checking if $U^\prime_n$ is $K_{n - 1}$ with one vertex of degree $1$ adjacent to it:} First, we count the degrees of the vertices. If there is agreement with the expected number of vertices of each degree, then we are done, since all vertices with degree $n - 2$ must form a $K_{n - 2}$ subgraph, all adjacent to the vertex of degree $n - 1$, since the vertex of degree $n - 1$ is already adjacent to the vertex of degree $1$. If this test passes and $U^\prime_n = H$, then the number of $k$-cliques is $\displaystyle {{n - 1} \choose k}$. If $U^\prime_n = \overline{H}$, then the number of $k$-cliques is $0$. If this test fails, then we move to the next test.

\item \textit{Checking if $U^\prime_n$ is $K_{n - 2}$ with two isolated vertices:} It suffices in this case to check alignment with the expected degrees of the vertices. If the test passes and $U^\prime_n = H$, then the number of $k$-cliques is $\displaystyle {{n - 2} \choose k}$. If $U^\prime_n = \overline{H}$, then the number of $k$-cliques is $0$. If this test fails, then we move to the next test.

\item \textit{Checking if $U^\prime_n$ is $K_{n - 2}$ with two vertices of degree $1$ adjacent to each other:} First, we compute the degrees of the vertices and check that the vertices of degree $1$ are adjacent to each other. This forces the other $n - 2$ vertices to form an $n - 2$-clique. If this test passes and $U^\prime_n = H$, then the number of $k$-cliques is $\displaystyle {{n - 2} \choose k}$. If $U^\prime_n = \overline{H}$, then the number of $k$-cliques is $0$. If this test fails as well, and after all other tests, we know from our classification that
\begin{equation*}
\left| \textit{Aut} \left( H \right) \right| = O \left( \frac{n!}{n^3} \right).
\end{equation*}

\end{enumerate}

In all six cases, in $\tilde{O} \left( n^2 \right)$-time (depending on one's preferred model of computation), we can determine the appropriate classification if
\begin{equation*}
\left| \textit{Aut} \left( H \right) \right| = \omega \left( \frac{n!}{n^3} \right),
\end{equation*}
and also compute the number of $k$-cliques in $\tilde{O} \left( n^2 \right)$-time. If not, we can determine that 
\begin{equation*}
\left| \textit{Aut} \left( H \right) \right| = O \left( \frac{n!}{n^3} \right).
\end{equation*}

\end{proof}
\end{document}